\documentclass[12pt]{article}
\makeatletter
\def\thanks#1{\protected@xdef\@thanks{\@thanks  \protect\footnotetext{#1}}}
\usepackage[a4paper, margin=1.2in]{geometry}
\usepackage{amsfonts}
\usepackage{amsthm}
\usepackage{amsmath}
\usepackage{mathtools}
\usepackage{dsfont}
\usepackage{authblk}
\usepackage{xcolor}
\usepackage{hyperref}
\usepackage{tabularx}
\usepackage{algorithm}
\usepackage{algpseudocode}
\usepackage{algcompatible}
\usepackage{longtable}
\usepackage{enumitem}
\usepackage{todonotes}
\setlist{noitemsep}

\usepackage [
  n,
  advantage,
  operators,
  sets,
  adversary,
  landau,
  probability,
  notions,
  logic,
  ff,
  mm,
  primitives,
  events,
  complexity,
  asymptotics,
  keys
  ] {cryptocode}

\usepackage{subcaption}

\usepackage{float}
\floatstyle{ruled}

\usepackage[english]{babel}

\newfloat{protocol}{H}{idf}
\floatname{protocol}{Protocol}

\newfloat{resource}{H}{idf}
\floatname{resource}{Resource}

\newfloat{simulator}{H}{idf}
\floatname{simulator}{Simulator}

\newfloat{reduction}{H}{idf}
\floatname{reduction}{Reduction}

\newfloat{recap}{H}{idf}
\floatname{recap}{Recap}

\newcommand{\ket}{\rangle}

\renewcommand{\sample}{\xleftarrow{\$}}

\newtheorem{theorem}{Theorem}

\newtheorem{lemma}{Lemma}
\newtheorem{remark}{Remark}
\newtheorem{definition}{Definition}
\newtheorem{proposition}{Proposition}

\title{Selectively Blind Quantum Computation}
 \date{}
 
\author[1,$*$]{Abbas Poshtvan\thanks{$^*$These authors contributed equally.}}
\author[2,$*,\dagger$]{Oleksandra Lapiha\thanks{$^{\dagger}$Work mostly done while at LIP6, Sorbonne Université.}}
\author[1]{Mina Doosti}
\author[1]{Dominik Leichtle}
\author[3]{Luka Music}
\author[1,4]{Elham Kashefi}

\affil[1]{\small School of Informatics, University of Edinburgh, UK}
\affil[2]{\small Royal Holloway, University of London, UK}
\affil[3]{\small Quandela, Massy, France}
\affil[4]{\small LIP6, CNRS, Sorbonne Université, Paris, France}

\begin{document}

\maketitle

\begin{abstract}
Known protocols for the secure delegation of quantum computations from a client to a server in an information-theoretic setting require quantum communication. In this work, we investigate methods to reduce the communication overhead. First, we establish an impossibility result by proving that server-side local processes cannot decrease the quantum communication requirements of secure delegation protocols. We develop a series of no-go results that prohibit such a process within an information-theoretic framework. 

Second, we present a possibility result by introducing Selectively Blind Quantum Computing (SBQC), a novel functionality which allows the client to hide one among a known set of possible computations. We precisely characterise how the differences between the computations in the protected set influence the number of qubits sent during our SBQC implementation, therefore yielding a communication-optimal protocol. This approach can reduce qubit communication drastically and demonstrates the trade-off between information leaked to the server and the protocol's communication cost. 
\end{abstract}

\section{Introduction}

\subsection{Context}

Currently, in order to delegate a quantum computation to a quantum service provider, a client sends over the input and the computation they want to perform, and the quantum server is supposed to do the computation and send back the outcome. For some important computations, the client may want the following privacy and security guarantees: the input, computation, and the output should stay unknown to the server (blindness); furthermore, the client should be able to check that the computation has been performed correctly (verifiability). Delegating quantum computations in this secure way to a powerful server is a crucial step toward making quantum computing widely accessible and establishing a secure quantum cloud. Hence, improving the efficiency and practicality of these functionalities together with the fundamental limitations on them has been a long-lasting topic of interest in quantum cryptography, complexity theory and physics~\cite{broadbent2009universal, drmota2024verifiable, childs2001secure, fitzsimons2017unconditionally, mahadev2018classical, reichardt2013classical,badertscher2020security,barz2013experimental,barz2012demonstration,mantri2017flow,polacchi2023multi,perez2015iterated}. In this manuscript, we focus on the blindness property only.

The Universal Blind Quantum Computing (UBQC) protocol proposed by \cite{broadbent2009universal} achieves this functionality provided that the client can generate random single qubits and send them over a quantum channel to the server. However, the number of qubits that must be sent in this protocol (i.e. its communication complexity) is linear in the size of the circuit. From a practical point of view, the flying qubit from client has to be integrated into the remote quantum computer of the server. In the process, it might be affected by the noise of the channel established between the parties, which results in a reduction of the computation performance on the server's side. The main question addressed in this paper is whether techniques can be used to reduce this communication complexity and also establish lower bounds on the number of encrypted qubits necessary to achieve blindness.

There are three existing approaches for constructing the blind quantum computation functionality, each using a specific resource: 
\begin{enumerate}
    \item Blind quantum computing with Prepare-and-Send \cite{broadbent2009universal}. This approach has been formulated in the measurement based quantum computation (MBQC) model \cite{raussendorf2001one}. Here, the client is able to apply random single qubit rotations and sends multiple encrypted qubits to the server (where the keys are the rotation angles). The server uses them to construct an encrypted graph state. The computation is performed interactively with the client sending measurement instructions and the server returning the measurement outcomes for subsets of the graph's qubits. At the end, the server sends back the output qubits to the client. The main resource for this variant of UBQC is the ability to prepare single qubits on the server's machine without the server knowing the state of these qubits, also called secure remote state preparation (RSP). This is the main variant of UBQC we will be considering throughout this paper, and will discuss it in more details later. A closely related alternative approach are Receive-and-Measure protocols \cite{morimae2013blind, childs2001secure}, in which the server prepares the computation graph state and sends it qubit-by-qubit to the client, who has the ability of measuring these qubits in a certain number of bases.

    \item  Blind quantum computing with non-communicating entangled servers \cite{reichardt2013classical}. This device-independent approach requires two non-communicating servers to share maximally entangled states with each other prior to the protocol's execution. Then the client instructs one of the servers to measure their part of this resource state and send back the outcome without the other server learning it. This information asymmetry is leveraged by the client to both perform the computation and test the servers' honesty by utilizing the rigidity of non local games \cite{palazuelos2016survey}. However, sharing high amounts of entanglement between distant servers and making sure that the severs are not communicating (e.g. by enforcing space-like separations so the parties can not have causal effects on each other) are the downsides of this approach.

    \item Blind quantum computing with computational hardness assumptions and classical client \cite{mahadev2020classical}. The idea behind this approach is to perform secure remote state preparations using only classical public-key cryptography \cite{mahadev2018classical,cojocaru2019qfactory}, in particular Trapdoor Claw-Free Functions (TCF) based only the Learning with Errors problem. In this approach, quantum communication is removed entirely, at the cost of having to implement the TCFs complex and deep circuits on the server's quantum hardware. Compared to the other approaches which are information-theoretically secure, these protocols provide only computational security.
\end{enumerate}

Among these approaches, the first one seems closer to implementation, as it has already been demonstrated in the lab ~\cite{barz2012demonstration}. From a practical perspective, the main obstacle here is the quantum communication cost which is equal to the size of the graph state implemented on the server (i.e. proportional to the number of gates in the circuit that the client wants to run). Reducing the communication cost and computing its fundamental lower bounds is an important step toward making this approach more feasible.

Apart from the experimental significance of achieving UBQC with reduced qubit communication, the possibility of accomplishing this with minimal resources has important implications for complexity theory. The first question in this direction is whether it is possible to make the client fully classical while preserving information-theoretic security. As discussed, this can be done with computational security using TCFs. However, achieving blindness with a fully classical client in an information-theoretic context is unlikely, as it would lead to unlikely collapses in complexity classes. As shown in \cite{aaronson2017implausibility}, delegating Boson Sampling in such a way would imply the collapse of the polynomial hierarchy to the third level. This relates to the foundational question of verifying BQP computations with BPP verifiers, which remains an open question in the field. Therefore, studying the communication requirements of the \cite{broadbent2009universal} protocol can shed light into possible answers to this question.

An important question in that direction is: \emph{Can we construct more efficient RSPs with less quantum communication?} This question was positively addressed in the Random Oracle Model \cite{zhang2021succinct}. They demonstrate that UBQC can be done by sharing keys encrypted in quantum states, with a communication complexity that scales polynomially with the security parameter and remains independent of the computation's size. To this low number of keyed states, they apply a gadget-increasing protocol to generate enough encrypted states to perform the client's computation blindly. However, while the Random Oracle Model is a nice theoretical model, it is not practical and similar protocols based on concrete instantiations have not yet been developed.

From another point of view, for many practical applications, hiding all aspects of the computation, i.e. the graph, the input, the algorithm, etc, might be unnecessary. For example, in the federated machine learning protocol that uses delegated quantum computing \cite{li2021quantum}, the only sensitive parameters are the weights of a run, and the other information about their circuit can be publicly revealed. For many cases, hiding only the input and some other qubits affected by their measurement outcome suffices. In that case, it would be sufficient to hide only the input and its propagation through the circuit \cite{caro2024interactive}.

\subsection{Overview of the Main Results}

\paragraph{Questions addressed.} 
A more comprehensive understanding of the minimal resources needed for blind quantum computations can be obtained from the answers to the following open questions:
\begin{enumerate}
\item For now, one RSP call generates a single keyed quantum state. Ideally, we would like the server to be able to expand this state into multiple encrypted states as in \cite{zhang2021succinct}. \emph{Is it possible to design more efficient RSPs from similar server-side processes in the plain model?}\footnote{Note that this new functionality is not equivalent to entropy expansion, but a distribution or splitting process} This would then imply that a run of the UBQC protocol requires less quantum communication for while preserving the same information-theoretic security.
\item The full UBQC protocol perfectly hides all the computation angles in the graph, the flow, the inputs and the outputs -- full blindness. However, perhaps there are settings in which the client does not care about the secrecy of some of these values. \emph{Is the full UBQC protocol still necessary and, if not, what is the minimal amount of quantum communication that is required to selectively blind a subset of these parameters?}
\end{enumerate}

\paragraph{Impossibility results -- Section \ref{sec:impossibility}.} In the first part of our work, we prove a sequence of no-go results which rule out some scenarios for reducing communication complexity of a protocol for implementing UBQC with full blindness.

More specifically, we study the existence of a process that allows the RSP process to be bootstrapped: implemented on the server's machine, it takes one keyed qubit as input and outputs two (or more) keyed qubits with different keys. We prove in the information theoretic setting that such a functionality, either as an isometry \ref{lemma:no-sep-gadget-d}\ref{sec:alternative-proof-lemma-1} or as a general quantum channel \ref{lemma:no-dist-gadget-channel}, can not be achieved due to the basic laws of quantum mechanics, i.e. linearity. We also consider the foundational aspect of this result and show that it can not be reduced to the no-cloning theorem. In lemma \ref{lemma:no-sep-gadget-d} We also show that it is impossible to construct a protocol that has the same security guarantees as the original UBQC by only having access to a process which replicates some but not all of the encrypted states.
 
Next, we turn our attention to variants of the the UBQC protocol of \cite{broadbent2009universal} which rely on entangled states or a more limited set of states than the original one. We again show the impossibility of amplifying such resource states to decrease the communication complexity. We rule out the entangled version by analyzing the symmetry of bipartite entangled systems in quantum mechanics -- the equality of the Von Neumann entropy of subsystems \ref{lemma:no-dist-entangled}. This result fundamentally relies on the fact that entanglement cannot be increased in a bipartite system by local operations.  

We then show that relaxing the requirements of the amplification processes described above is also not enough to yield a secure UBQC protocol. Indeed, we exhibit explicit differential attacks on a protocol makes use of amplifiers which can approximately prepare the correct states, i.e. states that have correlated parameters in the classical description of each subsystem, but the correlation is supposed to be minimal \ref{lemma:no-go-correlated-angles}. All these results are encompassed in \ref{prop:ubqc-limitation}.    

Finally, we prove the existence of a lower bound on the amount of quantum communication needed for blind quantum computing protocols which hide the same information as the original UBQC protocol: the number of unknown states sent to the server (or more generally, the number of unknown independent parameters) must be superlogarithmic in the size of the computation \ref{lemma:ubqc-lower-bound}.


\paragraph{Possibility results -- Section \ref{sec:selectively-blind-quantum-computing}.} In the second part of the paper, guided by the impossibility results above, we focus on optimizing the UBQC protocol by going back and analyzing the functionality that it achieves. We introduce a more general blind functionality (Resource \ref{res:ubqc short}) which we call Selectively Blind Quantum Computing (SBQC):  \emph{It allows the client to delegate one computation from a known set of unitaries to the server while hiding the choice itself.} UBQC then corresponds to the subcase of SBQC when the set of unitaries contains all the computations that are possible on a given graph. We then show that we only need to mask the differences between computational paths rather than the entire computation. SBQC demonstrates its full advantage when unitary operations differ minimally — for instance, by just one or a few gates. The protocol uses two main techniques:

\emph{Section \ref{sec:angle_masking} -- Angle masking:} We study here the hiding cost generated by a node in the computation graph which has different default measurement angles in the known computations. This is done by first defining the notion of a \emph{future cone} (Definition \ref{def:future_cone}) of a given node: all the nodes in the graph which are affected by -- i.e. which can get corrections directly or indirectly from -- the measurement outcome of that node. We then show that if only the measurement angles of some nodes are different between the graphs and all the other properties of the computations are the same, it is sufficient to hide with quantum communication (i) those specific nodes and (ii) the nodes in their future cones which have non Clifford measurement angles, i.e. odd multiples of $\pi/4$.

\emph{Section \ref{sec:graph_masking} Graph masking (Merge-and-Break):} If the computation graphs have different shapes, we give conditions for finding a \emph{merger graph} (Definition \ref{def:merger_graph} and Lemma \ref{lem:merger_graph}) that contains all computational graphs. We then introduce methods to decompose the graph into the desired configuration while concealing its shape, size, and flow. We achieve this by leveraging bridge-and-break operations \cite{MPKK18}, alongside established techniques from Measurement-Based Quantum Computing \cite{briegel2009measurement}.

We then combine these two techniques into a single full protocol for SBQC for delegating one out of two computations (Protocols \ref{prot:full masking cp} and \ref{prot:full masking pp}). This protocol can be easily generalized to arbitrary sets of computations by applying similar rules to combine all the graphs and measurement angles one-by-one. We finally prove that it is information-theoretic secure in the composable Abstract Cryptography framework (Theorem \ref{thm:blind}), meaning that it remains secure even if run sequentially or concurrently with other protocols.

Thus, by relaxing full universal blindness to SBQC, we significantly reduce communication overhead while maintaining robust information-theoretic security. Applying SBQC to practical scenarios such as algorithm auditing, servers can validate algorithmic structures without accessing hidden user data or selected computational choices. This approach effectively bridges theoretical UBQC with practical deployment within emerging distributed quantum computing architectures.

\paragraph{Note.} We note that a result similar to our angle masking technique, stating that hiding only non-Clifford gates is sufficient, has been previously known in the circuit model \cite{broadbent2015delegating}. We have rediscovered this fact within the framework of Measurement-Based Quantum Computation and utilize it as a subroutine to build a complete protocol, SBQC. \\ 
Additionally, a somewhat similar objective to SBQC was recently pursued in \cite{lee2025} to create a protocol that allows for concealing specific or sensitive parts of a quantum circuit. This is achieved by employing quantum fully homomorphic encryption (QFHE) at the input and output, while also hiding the sensitive portion of the circuit, such as the oracle in the Grover algorithm’s search function. While both approaches share a goal of selective concealment, our SBQC protocol employs entirely distinct techniques. Furthermore, our work provides a clearly defined functionality for this selective hiding and includes a comprehensive analysis of the communication cost increases that may arise from variations among the unitaries in the publicly known set, offering a detailed perspective.

\section{Preliminaries}
\label{sec:prelim}

In this section we briefly introduce the notions and technical tools that we will use throughout the paper.

\subsection{Notations}

We list the necessary Quantum Computing notations we use in this paper. The bra-ket notation is used to represent pure quantum states, such as $|\psi \ket$. The symbol $|0\ket$ represents the quantum state equal to $\binom{1}{0}$ in the vector notation. Then $|1\ket$ equals $\binom{0}{1}$ and $\alpha|0\ket + \beta|1\ket$ where $\alpha, \beta \in \mathbb{C}$ and $|\alpha|^{2} + |\beta|^2 = 1$ equals $\binom{\alpha}{\beta}$ in the vector notation. Specifically, we widely use the following quantum state parametrised by an angle $\theta$ as: $|\pm_\theta\ket = \frac{1}{\sqrt{2}}(|0\ket \pm e^{i\theta}|1\ket)$. We also use $\rho$ to represent the density matrix of an arbitrary qubit. Quantum systems are composed via tensor products. We use $\rho ^{\otimes n}$ to show the tensor product of $n$ copies of $\rho$


A unitary operation describes the reversible evolution of quantum systems. A unitary matrix $U:L(\mathbb{C}^d) \longrightarrow L(\mathbb{C}^d) $ is a linear operator where the input and output have the same dimension and $U  U^{\dagger}= U^{\dagger}  U=\mathbb{I} $. Also, an isometry $D$ in this manuscript refers to an operator which is inner product preserving, but in a more general sense than unitary operations. In an isometry $D:L(\mathbb{C}^{d_1}) \longrightarrow L(\mathbb{C}^{d_2})$ the input and output dimensions are different and we have $D^{\dagger} D =\mathbb{I}$. In other words, isometry is a generalization of a unitary where you can tensor product your input state with another state and perform a unitary on the whole system. The most general mathematical model for describing a quantum process is a quantum channel, $\Lambda:L(\mathbb{C}^{d_1}) \longrightarrow L(\mathbb{C}^{d_2})$ which is a completely positive and trace preserving (CPTP) map that maps the initial state of a quantum system to its final state. It can be shown that the action of a quantum channel on an input state can be written as a partial trace on the evolved state of the input by an isometry \cite{stinespring1955positive}:
\begin{equation}
 \forall \ \Lambda \  \exists \  D,E \  s.t. \quad  \Lambda(\rho)= tr_E (D \rho D^{\dagger})
\end{equation}
Where $E$ is the space of the tensor product state with the input $\rho$. 

One of the most widely used set of unitary operations are Pauli matrices, also known as Pauli gates, which form a basis for the operations over a single qubit system, described as follows:
\begin{equation}
\mathbb{I}, \quad X = \begin{pmatrix} 0 & 1 \\ 1 & 0 \end{pmatrix}, \quad Z = \begin{pmatrix} 1 & 0 \\ 0 & -1 \end{pmatrix}, \quad Y = \begin{pmatrix} 0 & -i \\ i & 0 \end{pmatrix}
\end{equation}
Another important gate is the Hadamard gate, which transforms the computational basis (Z) to plus-minus
basis (X), and has the following matrix form.
\begin{equation}
H = \frac{1}{\sqrt{2}}\begin{pmatrix} 1 & 1 \\ 1 & -1 \end{pmatrix}
\end{equation}
Also, an arbitrary rotation around the Z-axis is given by
\begin{equation}
    Z(\theta) = \begin{pmatrix} 1 & 0 \\ 0 & e^{i\theta} \end{pmatrix}
\end{equation}

Finally, we introduce a two-qubit controlled operation gate known as $CZ$. This gate applies a $Z$ gate on the target qubit if the control qubit is in state $|1\ket$, and applies an identity if it is in state $|0\ket$. The $CZ$ gate can act as an entangling operation as is given by the following matrix:
\begin{equation}
    CZ = \begin{pmatrix} 1 & 0 & 0 & 0 \\ 0 & 1 & 0 & 0 \\ 0 & 0 & 1 & 0 \\ 0 & 0 & 0 & -1  \end{pmatrix}
\end{equation}


\subsection{Measurement-Based Quantum Computing.}
\label{subsec:mbqc}
Quantum computation can be described in various ways universally. In this paper, we use a universal model known as Measurement-Based Quantum Computation (MBQC)~\cite{mbqc}. This model is well-suited for delegating quantum computation and the desired properties of it.


In MBQC, the computation is defined with a pattern (also called a resource graph) which is a set of qubits prepared in $|+\ket$ states, where some of them are entangled using a $CZ$ operator. It is usually represented by a graph $G$ where only entangled nodes are connected with an edge.
\[
    |G\ket = \prod_{(i,j) \in E_G} CZ_{(i,j)} |+\ket^{\otimes V_{G}}
\]
If the computation has quantum input, a subset of the nodes of $G$ is reserved for the input states, and another subset is reserved for the output states. We call such a graph $G$ an open graph.

\begin{definition}[Open graph]
    An open graph is a tuple $(G, I, O)$, where $G= (V,E)$ is an undirected graph and $I, O \subseteq V$ are the sets of input and output vertices, respectively.
\end{definition}


To compute a general unitary computation in MBQC, one needs to measure all non-output nodes sequentially. The set of universal gates can be implemented only using the measurements in the XY-plane of the Bloch sphere and the bases $|+_{\theta}\ket, |-_{\theta}\ket$ where the angles $ \theta  \in \{\frac{k\pi}{4}: k =0, \dots, 7\}$ suffice for an arbitrary quantum computation~\cite{P05}. In what follows whenever measuring with an angle $\psi$  is mentioned, we refer to measuring in the basis $| +_{\psi} \ket$, $| -_{\psi} \ket$.

As the measurements are probabilistic by nature, in the MBQC picture, to obtain a deterministic\footnote{We call a computation pattern deterministic if its final outcome does not depend on intermediate measurement results.} outcome, the intermediate measurement outcomes need to be taken into account and corrected at the end. It has been shown in~\cite{BKMP2007} that there exist a necessary and sufficient condition on an open graph known as \emph{g-flow} to achieve such deterministic computation. This condition is defined formally as:

\begin{definition}[G-flow]\label{def:glow}
    For an open graph $(G,I,O)$ a g-flow is a tuple $(g, \prec)$, where $g: V\setminus O \rightarrow 2^{V \setminus I} \setminus \{\emptyset\}$ and $\prec$ is a partial order on $V$, satisfying the following conditions:
    \begin{itemize}
        \item If $j \in g(i)$, then $i \prec j$.
        \item If $j \in \operatorname{Odd}(g(i))$, then $i \prec j$.
        \item For all $i$ we have: $i \in \operatorname{Odd}(g(i))$.
    \end{itemize}
    In this context, $\operatorname{Odd}(K) = \{u: |N(u) \cap K| = 1 \pmod{2}\}$ for $K \subseteq V$ is the set of odd neighbours of $K$.
\end{definition}

Given a g-flow, the computation can proceed by correcting the measurement angles of each node depending on the previous measurement outcomes~\cite{mbqc}. A node $v$ of the resource graph propagates an X-correction to every node in $g(v)$ and a Z-correction to every odd neighbour of $g(v)$ (See an example in Figure~\ref{fig:gflow}). These corrections are equivalent to applying $X$ and $Z$ Pauli gates to the corresponding nodes before the measurement, but they can be incorporated into the measurement angles.


 We can also compute the ``inverse'' formula and obtain all the corrections that \textit{affect} a particular node. Then we have a formula for correcting the measurement angle of a node $w$:
\begin{equation}
\label{eq:corrections}
\begin{aligned}
    &\varphi'_w = (-1)^{s_X}\varphi_w + s_{Z}\pi
    &\text{ where } \quad
    &s_X = \sum_{j: w \in g(j)} s_j\text{,  } s_Z = \sum_{j: w \in Odd(g(j)), j \neq w} s_j.
\end{aligned}
\end{equation}
Finally, a measurement pattern is defined as below: \\ 
\begin{definition}[Measurement pattern]\label{measurement.pattern}
    A pattern in measurement based quantum computation is given by a graph $G(V,E)$, input and output vertex sets $I$ and $O$, a flow $g$ which induces a partial ordering of the qubits $V$, and a set of measurement angles $\{\phi_i\}_{i \in O}$ in the $X-Y$ plane of the bloch sphere.  
\end{definition}

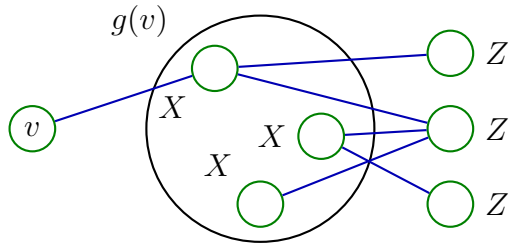
\begin{figure}[ht]
\centering
\begin{tikzpicture}
  \node[draw=black, thick, circle, minimum size=3cm, inner sep=0pt, label=above left:$g(v)$] (c) at (0,0) {};
  \node[draw=green!50!black, thick, circle, minimum size=0.6cm, inner sep=0pt] (v1) at (-3,0) {$v$};
  \node[draw=green!50!black, thick, circle, minimum size=0.6cm, inner sep=0pt, label=below left:$X$] (v2) at (-0.6,0.8) {};
  \node[draw=green!50!black, thick, circle, minimum size=0.6cm, inner sep=0pt, label=left:$X$] (v3) at (0.8,-0.1) {};
  \node[draw=green!50!black, thick, circle, minimum size=0.6cm, inner sep=0pt, label=above left:$X$] (v4) at (0,-1) {};
  \node[draw=green!50!black, thick, circle, minimum size=0.6cm, inner sep=0pt, label=right:$Z$] (v5) at (2.5,1.0) {};
  \node[draw=green!50!black, thick, circle, minimum size=0.6cm, inner sep=0pt, label=right:$Z$] (v6) at (2.5,0) {};
  \node[draw=green!50!black, thick, circle, minimum size=0.6cm, inner sep=0pt, label=right:$Z$] (v7) at (2.5,-1.0) {};
  \draw[blue!70!black, thick] (v1) -- (v2);
  \draw[blue!70!black, thick] (v2) -- (v5);
  \draw[blue!70!black, thick] (v2) -- (v6);
  \draw[blue!70!black, thick] (v3) -- (v6);
  \draw[blue!70!black, thick] (v3) -- (v7);
  \draw[blue!70!black, thick] (v4) -- (v6);
\end{tikzpicture}
\caption{In green we show all the nodes for which the measurement angle has to be adapted depending on the measurement result of $v$. The node $v$ induces an X-correction on every node in the set $g(v)$. It also induces a Z-correction on every odd neighbour of $g(v)$. Odd neighbours are the vertices that have an odd number of edges connecting them to $g(v)$.}
\label{fig:gflow}
\end{figure}

\subsubsection{Delegated computing in MBQC}
\label{subsubsec:delegated}

To perform a delegated computation in this framework between the Client with a small quantum device and the Server with a quantum computer, we proceed as follows. Firstly, the Client sends to the Server the qubits of the input and the classical description of the resource graph. Then the Server prepares a $|+\ket$ state for each non-input vertex and entangles prepared and received qubits with respect to the graph. After that, the Client sends to the Server a measurement angle for every non-output node. The Server measures them according to the g-flow and sends the measurement results to the Client. In the end, the Server sends to the Client the output qubits. These qubits can have extra X- and Z- Z-rotations because of the corrections. So the Client uncomputes them and recovers the final result. We note that performing delegation computation in this way, without any additional procedure, does not satisfy any security requirement, i.e. the server has full knowledge over the data and the details of the computations of the client. 

\subsection{Universal Blind Quantum Computation Protocol.}
\label{subsec:ubqc}

The delegated protocol explained in Section~\ref{subsubsec:delegated} does not provide any privacy for the client. However, there exist protocols for this model that can hide input, computation, and output of the client from the server. One protocol, known as Universal Blind Quantum Computation (UBQC), was first proposed by Kashefi et al. in~\cite{bfk}. We first describe informally how the protocol works. For a more detailed review, we refer the reader to \cite{fitzsimons2017private}.

Firstly, we implement the computation on a brickwork state to ensure that the graph always has the same structure and the same g-flow; for more details, see \cite{bfk}. To fully conceal the computation, it remains to hide the measurement angles. Consider a qubit in the resource graph that does not belong to the input or output sets; we refer to such a qubit as a \textit{computational} node. In the UBQC protocol, for each such node, the Client prepares a qubit in the state $|+_\theta \rangle$ with $\theta \in {\frac{k\pi}{4},; k = 0, \dots, 7}$ and sends it to the Server. The Server then entangles this qubit with the rest of the graph. When the time comes to measure the node, the Server requests a measurement angle from the Client. The Client responds with $\delta = \varphi' + \theta$. A short calculation shows that this measurement produces the same effect as measuring a $|+\rangle$ state with angle $\varphi'$, except that the angle sent to the Server is now hidden using a one-time pad. However, this masking alone introduces a subtle information leak. Suppose we measure a qubit $\rho$ with angle $\psi$ and obtain the result $s$. If $s = 1$, we can be certain that $\rho \neq |+{\psi} \rangle$; otherwise, we know $\rho \neq |-{\psi} \rangle$. Either way, we gain some information about $\rho$ from the outcome $s$. Since in our case $\rho = |+_{\theta} \rangle$, the measurement result $s$ can leak partial information about the secret key $\theta$. To prevent this, the measurement outcome of a masked node must also be hidden. The Client does this by sampling $r \sample {0,1}$ and sending the measurement angle $\delta = \varphi' + \theta + r\pi$. The term $r\pi$ flips the measurement basis when $r = 1$. The Server performs the measurement and returns the outcome $b$ to the Client. If the Server is honest, the Client can recover the true outcome as $s = b \oplus r$. Thus, the value of $s$ remains perfectly hidden. The input qubits to the computation are encrypted using $X^aZ(\theta)$. For these nodes, the Client sends a measurement angle of the form $\delta = (-1)^a\varphi + \theta + r\pi$. Because the Server does not know the actual measurement outcomes, the resulting $X$- and $Z$-corrections appear random to them. As a result, the output qubits are encrypted with a quantum one-time pad of the form $X^aZ^b$. The Client must apply the inverse of these corrections on their own device to recover the final output of the computation. We will now give a more formal description of the protocol. 

\begin{protocol}[ht]
\caption{Universal blind quantum computation~\cite{broadbent2009universal}}\label{ubqc-main}
\vspace*{1ex}
\begin{enumerate}
  \item \textbf{Alice's Preparation}  
  \begin{enumerate}
    \item For each column $x=1,...,n$ and each row $y=1,...,m$, Alice (the client) prepares the single-qubit state $|\psi_{x,y}\rangle$ as $|+_{\theta_{x,y}}\rangle$ where $\theta_{x,y} \in \{\frac{k\pi}{4} \ | \  k = 0 ,1 , ... ,7\}$, and sends the qubits to Bob (the server).
    \end{enumerate}
   \item \textbf{Bob's Preparation}
   \begin{enumerate}
       \item Bob creates an entangled state from all received qubits, according to their indices, by applying $CZ$ gates between the qubits in order to create a Brickwork state.
   \end{enumerate} 
    \item \textbf{Interaction and measurement} 
    \begin{enumerate}
    \item For each column $x=1,...,n$ and each row $y=1,...,m$:
    \begin{enumerate}
    \item Alice computes $\phi_{x,y} ^{\prime}$ where $s_{0,y} ^X = s_{0,y} ^Z =0$.
    \item Alice chooses $r_{x,y} \in \{0,1\} $ and computes $\delta_{x,y} = \phi_{x,y} ^{\prime} + \theta_{x,y}+ \pi r_{x,y}$.
    \item Alice transmits $\delta_{x,y}$ to Bob. Bob measures in the basis $\{|+_{\delta_{x,y}}\rangle , |-_{\delta_{x,y}}\rangle\}$.
    \item Bob transmits the result $s_{x,y} \in \{0,1\} $ to Alice.
    \item If $r_{x,y}=1$, Alice flips  $s_{x,y}$, otherwise she does nothing. 
    \end{enumerate}
    \end{enumerate}
\end{enumerate}
\end{protocol}

In the language of Abstract Cryptography (see Section~\ref{sec:ac-framework}), the functionality and security of UBQC is captured by Figure~\ref{fig:ideal-ubqc}~\cite{DFPR14}.
\begin{figure}
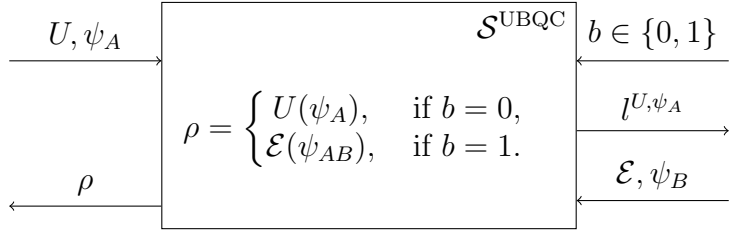

    \centering
    \begin{bbrenv}{resource}
    \begin{bbrbox}[name=$\mathcal{S}^\text{UBQC}$, minheight=3cm]
        $\rho = \left\{ \begin{matrix}
             U(\psi_A), & \text{ if } b=0, \\
            \mathcal{E}(\psi_{AB}), & \text{ if } b=1.
        \end{matrix} \right.$
    \end{bbrbox}
    \bbrmsgspace{6mm}
    \bbrmsgto{top={$U, \psi_A$}, length=2cm}
    \bbrmsgspace{12mm}
    \bbrmsgfrom{top={$\rho$}, length=2cm}
    \bbrqryspace{6mm}
    \bbrqryfrom{top={$b\in\bin$}, length=2cm}
    \bbrqryspace{2mm}
    \bbrqryto{top={$l^{U,\psi_A}$}, length=2cm}
    \bbrqryspace{2mm}
    \bbrqryfrom{top={$\mathcal{E},\psi_B$}, length=2cm}
    \end{bbrenv}
    \caption{The ideal functionality of UBQC, in the language of Abstract Cryptography.}
    \label{fig:ideal-ubqc}
\end{figure}

\subsection{Abstract Cryptography Framework}\label{sec:ac-framework}

There exist several formal frameworks to study and analyse the security properties of protocols such as UBQC. One of these frameworks is Abstract Cryptography (AC)~\cite{MR11,Mau12}. One of the advantages of AC over other formal ways, such as game-based security models, is its inherent composability: if a protocol is proven secure within the AC framework, it remains secure when composed, either sequentially or in parallel, with other secure protocols. For example, the confidentiality of a key established by a key exchange protocol holds even when multiple sessions run concurrently, provided each session satisfies the composable security definition. Another strength of AC lies in its modularity and clarity. The framework allows for precise specifications of what a protocol is intended to achieve (via ideal resources), as well as clear descriptions of the adversarial capabilities permitted in the model. This separation of concerns helps avoid ambiguity in security definitions and facilitates reasoning about complex protocol interactions.

At the core of the AC framework are resources — abstract entities (or ``boxes") that encapsulate a functionality and interact with external parties via interfaces. These resources are assumed to have unbounded computational power. When a party is honest, it interacts with the resource through a converter, a predefined, efficient CPTP (completely positive trace-preserving) map that mediates input/output behavior. Dishonest parties, on the other hand, can replace their converters with arbitrary CPTP maps. Certain interfaces may be filtered, meaning they are accessible only to malicious parties and are used to model adversarial control points. Capturing all the aspects of the protocol in such ideal resources is crucial and one of the sources of complications in AC. If the ideal resource fails to capture all relevant adversarial behaviors, the resulting security guarantee may be incomplete or even misleading. This places a strong burden on the designer to carefully and comprehensively define the security functionality.

Security in AC is formalized via resource construction. We say that a protocol $\pi$ $\epsilon$-constructs an ideal resource $\mathcal{S}$ from a real resource (or collection of resources) $\mathcal{R}$, written as:
\begin{equation}
    \pi \mathcal{R} \approx_{\epsilon} \mathcal{S}
\end{equation}

This expression means that the behavior of the protocol execution using real resources is $\epsilon$-indistinguishable from the ideal resource $\mathcal{S}$, according to a defined metric. As such, the resource construction can be defined as:

\begin{definition}[Construction of Resources]
    Let $\epsilon(k)$ be a function of the security parameter $k$. We say that a two party protocol $\pi = (\pi_1, \pi_2)$ $\epsilon$-statistically-constructs a resource $\mathcal{S}$ from resource $\mathcal{R}$ if:
    \begin{itemize}
        \item It is correct: $\pi_1\mathcal{R}\pi_2 \approx_{\epsilon} \mathcal{S}$.
        \item It is secure when the second party is corrupted: there exists a simulator $\sigma_{2}$ such that $\pi_1 \mathcal{R} \approx_{\epsilon} \mathcal{S}\sigma_{2}$.
    \end{itemize}
\end{definition}

To formalize indistinguishability, we introduce a Distinguisher, an entity whose goal is to decide whether it is interacting with the real protocol $\pi \mathcal{R}$ or the ideal resource $\mathcal{S}$. The Distinguisher may choose all inputs to the system, simulate the behavior of corrupted parties, and observe all outputs. The nature of the Distinguisher’s power determines the type of security achieved: if the Distinguisher is computationally unbounded, we obtain statistical security; if it is limited to efficient computations, the result is computational security.

\begin{definition}[Statistical Indistinguishability of Resources]
Let $\epsilon(k)$ be a function of security parameter $k$ and let $\mathcal{R}_1$ and $\mathcal{R}_2$ be two resources with the same input and output interfaces. The resources are $\epsilon$-statistically indistinguishable if, for all distinguishers $\mathcal{D}$, we have:
\[
    \big| Pr(b=1|b \leftarrow \mathcal{DR}_1) - Pr(b=1|b \leftarrow \mathcal{DR}_2) \big| \leq \epsilon
\]
Then we write $\mathcal{R}_1 \approx_{\epsilon} \mathcal{R}_2$
\end{definition}

To prove the indistinguishability between the real-world protocol execution and the ideal resource, we often construct an efficient Simulator. The Simulator connects to the ideal resource via the interface of the corrupted (malicious) party and exposes an additional interface that interacts with the Distinguisher. The role of the Simulator is to mimic the behavior of the real-world adversary in such a way that the Distinguisher cannot tell whether it is interacting with the real protocol execution or with the ideal resource and the Simulator. In other words, the Simulator must generate interaction patterns, based only on the limited information available through the ideal resource, that are indistinguishable from those produced in the real world. The overall goal is to show that for any Distinguisher, there exists a Simulator such that the joint system composed of the Simulator and the ideal resource produces an output distribution that is $\epsilon$-close (in the appropriate metric) to the output of the real protocol execution. This guarantees that the ideal resource is a faithful abstraction of the protocol's security in the presence of adversaries.

\section{Impossibility Results: No randomness distribution}
\label{sec:impossibility}

In this section, we study the possibility of implementing a local quantum process on the server side that takes as input one of the $|+_{\theta}\rangle$ states and outputs two such states (or more general states but in the product form or even entangled form with distinct independent parameters) which could serve as the vertices of the ultimate graph state we want to run the computations on. We name this machine a remote state expander (RSE). \\
If such a process exists, it can clearly reduce the communication cost in blind quantum computing. We will be studying the existence of such a machine in the information theoretic setting and in the standard model. We show that the separable output variant of this machine can not exist due to the linearity of quantum mechanics, which is also the root of no cloning theorem, a well known and fundamental theorem in quantum mechanics \cite{wootters1982single}. Finally in the last subsection, we provide a counting argument that puts a lower bound on the number of unknown parameters to the server, indicating a limitation of reducing the quantum communication complexity. Let us start by stating our main no-go results in the following proposition:

\begin{proposition}[No-go on Communication Complexity reduction for UBQC]\label{prop:ubqc-limitation}
    In any UBQC protocol between a client and a server, there exists no local quantum operation $\Lambda_S$ (prescribed by the client and performed by the server) mapping a resource state $|+_{\theta}\ket$ to a state $\rho$ used on two or multiple nodes of the underlying graph state, such that the new resource, satisfies the completeness and blindness of the protocol simultaneously. In particular, no such server-side process can generate or recycle resource states in a way that reduces the quantum communication from the client without compromising security.
\end{proposition}


We support our claim in the above proposition through a central, fundamental no-go theorem and a lemma concerning UBQC. Our first no-go result rules out the possibility of a universal gadget for RSE. Nonetheless, this no-go result is not limited to the scope of UBQC and delegated computation; more generally, it demonstrates that no universal quantum machine exists for distributing a random resource state into multiple useful resource states related to the original random resource. We begin by stating the theorem, and to capture this no-go in full generality, we provide a sequence of lemmas that establish the result across all relevant scenarios.

\begin{theorem} \label{th:main-no-go}
    There exists no quantum machine $D:L(\mathbb{C}^2) \longrightarrow L((\mathbb{C}^2)^{\otimes2})$ that takes as input a state $|+_{\theta}\ket$, and maps it into a bi-partite system $\rho_{A,B}$ as $D (|+_{\theta}\rangle) = \rho_{A,B} (\theta_1, \theta_2)$, such that all the following three conditions hold:
    \begin{equation} \label{ent}
    \begin{split}
        & \text{(1) } tr_B(\rho_{A,B} (\theta_1 , \theta_2)) = \rho_A(\theta_1) \\
        & \text{(2) } tr_A(\rho_{A,B} (\theta_1 , \theta_2)) = \rho_B(\theta_2) \\
        & \text{(3) } \text{If the parameter $\theta$ is unknown, then no information about $\theta_1$ and $\theta_2$ is accessible}\\
        & \text{from the output $\rho_{A,B}$.}
    \end{split}
\end{equation} 
\end{theorem}
\noindent The proof of the theorem is given after Lemma~\ref{lemma:no-dist-entangled}. The above general no-go has a particularly important consequence in the context of UBQC, where the instruction for the description of $D$ comes from the client and is implemented by the server. If such a gadget $D$ were to exist, the server could locally generate additional independent resource states from a given $|+_{\theta}\rangle$ state, where the only available information was the linear dependency equation between the input and output angles. These newly generated states could then be used in UBQC without compromising blindness. With such a gadget, it would be possible to reduce the quantum communication overhead. However, since such a linear operation is quantum mechanically impossible, no general strategy can include such a resource optimisation subroutine. This significantly limits the possibility of designing a more communication-efficient UBQC protocol. Also note that here $D$ is any general linear operation. More specifically, in the following lemma, we assume $D$ to be an isometry and we show that no such isometry can exist to provide a separable structure, as
\begin{equation}\label{eq:d-gadget}
    D (|+_{\theta}\ket) = |+_{\theta_1} \rangle \otimes |+_{\theta_2} \rangle \hspace{0.5cm}, \forall \  \theta.
\end{equation}
where condition $3$ of Theorem~\ref{th:main-no-go} is satisfied. i.e. $\theta_1$ and $\theta_2$ cannot recovered if $\theta$ is fully unknown.

\begin{lemma}\label{lemma:no-sep-gadget-d}
    There exists no isometry operation $D:L(\mathbb{C}^2) \longrightarrow L((\mathbb{C}^2)^{\otimes2})$ outputting according to~\ref{eq:d-gadget}, while no information about $\theta_1$ and $\theta_2$ is accessible, if $\theta$ is unknown. 
\end{lemma}
\begin{proof}
We give a constructive proof here, and we also provide a second general proof by contraposition in Appendix~\ref{sec:alternative-proof-lemma-1}. The second proof can also be seen as a no-cloning type of proof for random state distribution by any quantum operation.

We first note that for this lemma, we explicitly require the output to satisfy a separable product from. Although $D$ itself can be a separable or an entangling process. We consider each case separately, starting from the separable case.  Also, the isometry can take any arbitrary ancillary system as an input. Without loss of generality, we can assume this ancillary state is instantiated in state $|+_{\theta'}\ket$ (as $D$ is arbitrary, it can include mapping the initial ancillary state into any arbitrary state). Let's assume now that $D$ consists of a separable one-qubit unitary as: 
\begin{equation}
    D(|+_{\theta}\rangle) = U_1 \otimes U_2 \  (|+_{\theta}\rangle \otimes |+_{\theta'}\rangle) 
\end{equation}
where $U_1, U_2: L(\mathbb{C}^2) \longrightarrow L(\mathbb{C}^2) $ are the single qubit operations. With no mixing between the states of the two input qubits, the description of $U_1$ and $U_2$ fully characterizes at least one of the two output angles. Thus, clearly no separable operation can satisfy condition $(3)$.



We then conclude that such $D$ needs to include entangling operations. We ask if we can construct an isometry which outputs two product states of the desired form in Equation~\ref{eq:d-gadget} for all input angles $\theta$. Assume that such an isometry exists. We write its action on a full basis and without loss of generality, we select the Hadamard basis: 
\begin{equation}
\begin{split}
    & D(|+\rangle) = |+_{\alpha} \rangle \otimes |+_{\beta} \rangle,\\
    & D(|-\rangle) = |+_{\gamma} \rangle \otimes |+_{\delta} \rangle
\end{split}
\end{equation}

Using rotational symmetry, and given by two local operations on each qubit, $|+\rangle$ can be mapped to $|+\rangle \otimes |+\rangle$, the number of parameters of $D$ is reduced by two, and we can rewrite a new isometry $D_2$ such that:
\begin{equation}
    D_2 = (Z (-\alpha) \otimes Z (-\beta)) D,
\end{equation}
\noindent where $D_2$ acts as $D_2(|+\rangle) = |+\rangle \otimes |+\rangle $. As the isometry should preserve the inner products implies that: 
\begin{equation}
   \langle +| - \rangle =  (\langle + |D_2^{\dagger}) (D_2 |-\rangle) = \langle + , + | +_{\gamma} , +_{\delta}\rangle = 0 
\end{equation}
which forces one of the remaining parameters, $\gamma$ or  $\delta$, to be equal $\pi$. Let us choose $\gamma =\pi$, leading to $D_2$ satisfying:
\begin{equation} \label{action on basis}
    D_2(|+\rangle) = |+\rangle \otimes |+\rangle, \quad D_2(|-\rangle) = |-\rangle \otimes |+_{\delta}\rangle,
\end{equation}
Using the above equation, we can easily derive the matrix form of this isometry in the Hadamard basis to be the following: 
\begin{equation}
    D_2 = \left(\begin{array}{cc}1 &0\\ 0&0 \\ 0 & cos(\delta) \\ 0 & sin(\delta)\end{array}\right)
\end{equation}
We now require the action of such an isometry on a $|+_{\theta}\rangle$ to satisfy the desired tensor product form. Writing the planar input state as $|+_{\theta}\rangle = \frac{1}{2} ((1+e^{i\theta})|+\rangle + (1-e^{i\theta})|-\rangle)$, we have:
\begin{equation}
     D_2 (|+_{\theta}\rangle) = \frac{1}{2} \left(\begin{array}{cc}1+e^{i\theta} \\ 0 \\ cos \delta (1-e^{i\theta})  \\ sin \delta (1-e^{i\theta})\end{array}\right)  
\end{equation}

The separability condition implies that:
\begin{equation}
    (1+e^{i\theta})(1-e^{i\theta}) \sin \delta = 0 
\end{equation}
The possible solutions to the above equation are: either $\theta=0$ and $\theta= \pi$, or $\sin \delta =0$ which leads to $\delta= 0, \pi$. In $\delta=\pi$, the problem reduces to quantum cloning, and hence impossible. The $\delta=0$ case reduces to the separable case, which we discarded earlier, where the condition $(3)$ cannot be achieved. This means that no such isometry can exist for all $\theta$, and the proof is complete.
\end{proof}

\noindent We concluded that no reversible linear operation exists for our desired gadget. We note that this lemma not only applies to the original UBQC protocol, but also to extensions that only require non-identical states with zero overlaps as shown in~\cite{dunjko2016blind}. This is due to the fact that we put no restrictions on $\theta_1$ and $\theta_2$ apart from the blindness requirement captured by condition $(3)$.
Now we generalise the no-go to quantum channels, by relaxing the purity. We are still requiring the tensor product form of the output. We state the no-go in the following lemma:  

\begin{lemma}\label{lemma:no-dist-gadget-channel}
  There exists no quantum channel $\Lambda: L(C^2) \longrightarrow L(C^2)$, satisfying the following output separability, while no information about $\theta_1$ and $\theta_2$ is accessible, if $\theta$ is unknown.
\begin{equation}\label{eq:channel-generlization}
    \Lambda (|+_{\theta}\rangle \langle +_{\theta} | ) = \rho (\theta_1) \otimes \rho (\theta_2). 
\end{equation}  
\end{lemma}

\begin{proof}
We assume such a CPTP map exists. We use the Stinespring dilation theorem~\cite{stinespring1955positive} to represent the channel via an isometry $G$, such that:
\begin{equation}
    \Lambda(|+_{\theta}\rangle \langle +_{\theta} |) = tr_{3,4} (G |+_{\theta}\rangle \langle +_{\theta} | G^{\dagger})    
\end{equation}
More precisely the existence of $\Lambda$ implies the existence of an isometry $G: L(H_A) \longrightarrow L(H_B \otimes H_D \otimes H_C \otimes H_E ) $ such that $G |+_{\theta}\rangle= |\psi_{\theta_1}\rangle \otimes |\psi_{\theta_2}\rangle$ where $\rho(\theta_1)= tr_D (|\psi_{\theta_1}\rangle \langle \psi_{\theta_1}| )$ and $\rho(\theta_2)= tr_E (|\psi_{\theta_2}\rangle \langle \psi_{\theta_2}| )$. However, in Lemma~\ref{lemma:no-sep-gadget-d} we showed no such isometry exists. Hence, we reach a contradiction and conclude the proof by contraposition.
\end{proof}

We now relax the separability condition. We  assume an RSE machine which outputs an entangled state with two parameters, such that the reduced density matrix of each subsystem is only a function of one of the parameters. Even though, to the best of our knowledge, no UBQC protocol has been designed working with such entangled states, we argue that it might still be possible to do so. However, the independence of the two subsystems in terms of the secret parameter is the minimal requirement for such states to be considered a resource state for UBQC (we will clarify this later as well) Hence, proving a no-go under this assumption will rule out any information-theoretic strategy or subroutines of a new UBQC protocol that can be reduced to the creation of such resource states. In the next lemma, we give this no-go for the isometry, which can always be extended to the channel form using Steinspring. 



\begin{lemma}\label{lemma:no-dist-entangled}
There exists no isometry operation $D:L(\mathbb{C}^2) \longrightarrow L((\mathbb{C}^2)^{\otimes2})$ outputting an entangled state of two parameters $\theta_1, \theta_2$ as $|\psi_{\theta_1 , \theta_2}\rangle$, while no information about $\theta_1$ and $\theta_2$ is accessible, if $\theta$ is unknown, and the following condition is satisfied:
\begin{equation}
    \begin{split}
        & \rho_A(\theta_1) = tr_B (|\psi_{\theta_1 , \theta_2}\rangle \langle |\psi_{\theta_1 , \theta_2}| )\\
        & \rho_B(\theta_2) = tr_A (|\psi_{\theta_1 , \theta_2}\rangle \langle |\psi_{\theta_1 , \theta_2}| )
    \end{split}
\end{equation}\label{eq:entangled-state-ind-cond}
\end{lemma}
\begin{proof}
For the conditions of Equation~\ref{eq:entangled-state-ind-cond} to be satisfied, we need the subsystems $A$ and $B$ to have at least one independent pair of parameters. i.e. they can still be correlated via other parameters potentially included in the description of $D$. Since $|\psi_{\theta_1, \theta_2}\rangle$ is a pure bipartite state, due to the Schmidt decomposition, we know that the entropy of subsystems should always be equal:  
\begin{equation} \label{entropies-equality}
    S(\rho_A)= S(\rho_B) 
\end{equation}
where $S(\rho)= - tr(\rho \log \rho) $ denotes the Von Neumann entropy. For the independence condition to be satisfied, the left-hand side of \ref{entropies-equality} is only a function of $\theta_1$ (and not $\theta_2$), and vice versa for the right-hand side. 
Therefore, the entropy identity implies that:
\[
S(\rho_A(\theta_1)) = S(\rho_B(\theta_2)) =: s,
\]
for some constant $s$ independent of both $\theta_1$ and $\theta_2$. This means that for a given input, $S(\rho_A(\theta_1))$ must be constant for all values of $\theta_1$, and similarly for $\rho_B(\theta_2)$. But this contradicts the assumption that $\rho_A$ and $\rho_B$ encode the parameters $\theta_1$ and $\theta_2$ respectively, unless $\rho_A$ or $\rho_B$ are constant with respect to their associated parameters. This is not the case, as we assume the density matrices of the two subsystems are distinct and are characterisable by parameters $\theta_1$ and $\theta_2$, such that the change of the input state $\theta$ results in changing the output parameters. If the marginals carry no information about $\theta_1$ or $\theta_2$, then no information about these parameters is accessible from any subsystem.

This would imply that either the state $|\psi_{\theta_1,\theta_2}\rangle$ is a separable product state, which is in contradiction with the initial assumption that it is entangled, or that the entropy of the subsystems must change with respect to the change of one of the parameters, which contradicts the Schmidt decomposition, Therefore, no such isometry $D$ can exist.
\end{proof} 

\noindent We are now ready to give the proof of Theorem~\ref{th:main-no-go}.\\

\noindent \textbf{Proof of Theorem~\ref{th:main-no-go}.}  Through Lemma~\ref{lemma:no-sep-gadget-d}, we showed that no isometric map can satisfy all three conditions of Equation~\ref{eq:d-gadget} simultaneously for all input angles $\theta$ when restricted to pure product outputs. In Lemma~\ref{lemma:no-dist-gadget-channel}, we extended the analysis to general quantum channels, allowing mixed separable outputs, and showed that even under this relaxation, such a map remains impossible. Finally, in Lemma~\ref{lemma:no-dist-entangled}, we relaxed the separability condition and considered general entangled outputs, while maintaining condition (3) in the form of independence of output parameters. We demonstrated that no quantum-mechanically valid map satisfies all required conditions in this setting either. These results collectively rule out all possible constructions of $D$, establishing that no such map exists. \qed

So far, we have shown that generating more resource states with independent parameters on the server side is impossible. To complete our argument for justifying Proposition~\ref{prop:ubqc-limitation}, we need to show that the set of assumptions in Theorem~\ref{th:main-no-go} is minimal for UBQC. This means that relaxing the assumption of independence of the parameters will lead to states that cannot be used in UBQC. To this end, we assume that our RSE subroutine is an approximate random state distributor, producing minimally dependent outputs.

We describe such an approximate quantum machine as a process that takes an input state $|+_{\theta}\rangle$ and produces a two-qubit output state, where each subsystem depends on both output parameters. In the spirit of approximate quantum operations — such as approximate quantum cloning~\cite{cerf2000asymmetric,scarani2005quantum} — one can design such a machine to minimize the degree of parameter cross-dependence. While the study of these approximate quantum random state distributors is an interesting direction within quantum information theory, we do not explore their full characterization here and leave it for future work.

Instead, we focus on whether such approximate machines can be useful in the context of UBQC, particularly when their outputs are used to construct graph states, as a means to reduce the quantum communication cost. Suppose an approximate distributor is used to prepare qubits in a graph state $|G\rangle$, where certain angle parameters $\theta_i$ are functionally dependent on others (e.g., $\theta_i = f(\theta_j)$), and this dependency is known to the server. We analyze this setting in the pure-state separable case. We show that using such dependent outputs in UBQC compromises the protocol's blindness. Although the measurement outcomes remain secure due to the one-time pad, the encryption of measurement angles can leak information if the server leverages the known parameter dependency. Therefore, we rule out the use of such approximate subroutines in UBQC through the following lemma. 

\begin{lemma}\label{lemma:no-go-correlated-angles}
Let $|G\rangle$ be a graph state used in a UBQC protocol, where the state of each node is described by a reduced density matrix $\rho(\theta_i)$. Suppose there exist indices $i, j$ such that the angle $\theta_j$ is functionally dependent on $\theta_i$, i.e., $\theta_j = f(\theta_i)$ for some deterministic function $f$, and this dependency is known to the server, even through the values of $\theta_i$ remain hidden. Then, the UBQC protocol may not satisfy blindness.
\end{lemma}
\begin{proof}
We prove the lemma by showing that a UBQC protocol using such a graph state $|G\rangle$ cannot guarantee blindness. In particular, the blindness condition reduces to the security of a classical one-time pad with correlated keys, which is, in general, insecure. As a result, the server can gain non-trivial information about the client's encrypted measurement angles, thereby violating the blindness requirement.

UBQC protocols employ two layers of encryption: \emph{1) Encryption of Measurement Outcomes}, and \emph{2) Encryption of Measurement Angles}.

The first encryption operates as follows: for a measurement on qubit $i$, the outcome $s_i$ is one-time padded with a random classical bit $r_i$, resulting in $b_i = s_i + r_i$. The server receives $b_i$. This encryption remains secure regardless of any dependencies among the $\theta_i$, because $r_i$ is chosen uniformly at random, independently of $\theta_i$, and remains unknown to the server. Therefore, $b_i$ reveals no information about $s_i$.

The second encryption works as follows: the client instructs the server to measure qubit $i$ at angle $\delta_i = \phi_i + \theta_i + r_i \pi$. Here, the server receives $\delta_i$ and is assumed to know the dependency structure among the $\theta_i$, for example, $\theta_j = f(\theta_i)$. Since the $\theta_i$ serve as keys in this one-time pad, the security of the angle encryption reduces to the security of a one-time pad with correlated keys. It is well known in classical cryptanalysis that such schemes are vulnerable to differential attacks~\cite{biryukov2009related,filiol2020key}. For instance, if $\theta_j = f(\theta_i)$, the server can analyze the pair $(\delta_i, \delta_j)$. Given knowledge of $f$, the server can mount differential attacks to infer relationships between $\phi_i$ and $\phi_j$, leaking information about the underlying computation. Consequently, the second layer of encryption is compromised. More generally, no UBQC protocol relying on information-theoretic encryption can satisfy blindness when using a resource state $|G\rangle$ of the form described in the lemma. This completes the proof.
\end{proof}
The consequence of the above lemma is that, even if an optimal approximate RSE machine exists, it cannot be employed within a UBQC protocol to reduce communication resources, as it would compromise the security of the protocol. Furthermore, we observe that any attempt to reduce the communication overhead of UBQC without introducing additional assumptions, i.e., relying solely on information-theoretic constructions, while involving a local process by the server, must involve a subroutine that falls within the class of processes described in Theorem~\ref{th:main-no-go}. However, as proven, no such subroutine exists. Alternatively, such attempts would result in a graph state as described in Lemma~\ref{lemma:no-go-correlated-angles}, which violates the blindness condition, thus failing to preserve the security of UBQC. This leads to the conclusion stated in Proposition~\ref{prop:ubqc-limitation}.

We finish this section with a couple of remarks to point out some interesting conceptual aspects of our no-go results.\\ 

\begin{remark}
    \textbf{(Distinction from Entropy Expansion):} It is important to emphasize that the functionality we aim to achieve in Theorem~\ref{th:main-no-go} is not equivalent to entropy expansion, but rather to a distribution or splitting process of a random bitstring into different states. Thus, it cannot be trivially addressed by entropic arguments. In other words, the process does not seek to generate additional random bits from a bitstring encoded in $\theta$ without any extra resources. Instead, it splits the same amount of randomness into two states, with the outputs being non-trivially linked to the inputs. This distinction allows the process to remain consistent with entropic principles, thereby highlighting the non-triviality of our no-go result
\end{remark}

\begin{remark}
 \textbf{(Distinction from the No-Cloning Theorem):} The functionality described in Theorem~\ref{th:main-no-go} may remind us the no-cloning theorem~\cite{wootters1982single}. In this context, we are again considering a machine $C: L(H_A) \longrightarrow L(H_B \otimes H_C)$, which operates as a cloner, such that $C(\rho) = \rho \otimes \rho$. This raises the question of whether the no-distribution theorem we established earlier can be reduced to the no-cloning theorem. We will argue that this is not the case, and this distinction can be seen from two key perspectives:\\
\emph{1.} First, although both theorems are based on the linearity of quantum mechanics, the no-randomness distribution theorem is more general. In the proof of Lemma~\ref{lemma:no-sep-gadget-d}, we recover a cloning scenario for a fixed parameter $\delta = \pi$, while for a general range of the parameter, our machine is allowed to produce more general outputs other than clones. We can also refer to our second proof of Lemma~\ref{lemma:no-sep-gadget-d} in Appendix~\ref{sec:alternative-proof-lemma-1} for a similar argument based on dimensionality, which again subsumes the perfect cloning machines.\\
\emph{2.} One might attempt to reduce Theorem~\ref{th:main-no-go} to the no-cloning theorem by suggesting that if quantum process existed, it could be transformed into a cloning machine by applying local rotations, i.e., $Z(\theta - \theta_1) \otimes Z(\theta - \theta_2)$. Thus, the no-go on RSE would follow from the no-cloning theorem. However, this reduction fails for one key reason: the unknown nature of the input. An essential assumption in both our theorem and the no-cloning theorem is that the input state, which is intended to be cloned or distributed, is unknown. Consequently, a local rotation of the form $Z(\theta - \theta_1) \otimes Z(\theta - \theta_2)$ cannot be applied, since we do not know the value of the input angle $\theta$.\\
Therefore, while the no-randomness distribution theorem shares a common root with the no-cloning theorem, namely, the linearity of quantum mechanics (which can also be interpreted as the difficulty of disentanglement), it cannot be directly reduced to the no-cloning theorem, and can be seen as a new fundamental no-go result in quantum information.
\end{remark}

\begin{remark}
\textbf{(On the approximate randomness distributor):} The concept of an approximate randomness distributor, as motivated by the previous lemma, is not the primary focus of this work. However, it offers a promising direction for future research into the properties of quantum machines. Specifically, it refers to a machine that generates a bipartite entangled state, where the correlation between the subsystems is minimized. In this context, we examine the correlation between the parameters that each state depends on. Unlike the cloning case, where the objective is to maximize fidelity, the appropriate measure of correlation in this case is not straightforward. To quantify the functionality of an approximate randomness distributor, suitable measures of correlation, such as quantum mutual information \cite{wilde2013quantum} or quantum discord \cite{ollivier2001quantum}, could serve as potential candidates.
\end{remark}


\subsection{A Lower Bound on Communication Complexity of UBQC Based on a Counting Argument}
In the previous section, we showed an important limitation for reducing the quantum communication of UBQC through the impossibility of RSE. In this section, we use a counting argument to establish a lower bound on the number of qubits that must be communicated for a blind quantum computing protocol to remain secure. This can also be interpreted as a lower bound for any collective generalization of the randomness-distribution machine discussed earlier.\\


The lower bound is derived by considering a simple attack strategy that a malicious server could employ to compromise the security of a blind quantum computing protocol. Specifically, the server could randomly guess the input parameters $\{\theta_i\}_{i=1,2,...,m}$ (the keys) and then implement and interpret the corresponding computation based on those values. We aim to identify a necessary condition for ensuring the security of a UBQC protocol against this type of attack. To illustrate this concept, consider the case where the only input the server receives from the client is a single-keyed qubit $|+_{\theta}\rangle$. In this scenario, the server could randomly guess the parameter $\theta$ and proceed to implement and interpret the computation accordingly. In such a case, the server would correctly guess the computation with a probability of $\frac{1}{8}$, which is a non-negligible probability. We formalise this argument in the following lemma.

\begin{lemma}\label{lemma:ubqc-lower-bound}
The number of qubits $m$ needed to be communicated for a blind quantum computation is lower bounded by $\omega (\log n)$, where $n$ is the number of vertices of the underlying graph state.   
\end{lemma}
\begin{proof}
Let $n$ be the number of vertices in the underlying graph state for UBQC. Each of these vertices corresponds to a $|+_{\theta}\rangle$ state where $\theta \in \Theta$ and where $\Theta$ is a finite set. Assume a subset of these qubits are not blinded and known to the server, and some others are encrypted and sent by the client.  If the number of sent qubits from the client to the server is m, then there are $|\Theta|^m$ possibilities for the whole graph state, which we denote each by the set $\{G_i\}_{i=1} ^{|\Theta|^m}$. As the probability distribution on this is uniform when the server guesses one of them randomly, the probability of each of these graphs ($G_i$) to be the desired graph state ($G^*$) which the client wishes to be executed is exponentially decreasing in $m$, and it needs to be a negligible function of $n$, the size of the graph (proportional to the size of computation) for the protocol to stay secure against random guess attack, which means we should have: 
\begin{equation}
    Pr[G(i)= G^*] = \frac{1}{|\Theta|^m} = negl(n)
\end{equation}
Solving this asymptotic equality for $m$ easily yields $m =\omega (\log n)$.
\end{proof}

\noindent This argument can be generalized in two ways:  
\begin{enumerate}
    \item If the protocol uses other types of states that are not $|+_{\theta}\rangle$, but still contain $m$ independent parameters (e.g., a highly entangled GHZ state with $m$ rotated qubits, where the rotation angles of each pair of qubits are independent), the same argument holds. Specifically, the number of possible underlying graph states will be exponentially large ($|\Theta|^m$). This lower bound arises from the fact that, as the size of the graph increases, the dimensionality of the input parameter space must also increase to keep the success probability of the random guess strategy negligible with respect to the computation size.
    \item If the set $\Theta$ is not discrete but continuous (and symmetric, as it should be), the claim still holds approximately. In this case, the server can discretize the space of input states by dividing it into $|\Theta|$ regions, each with a length of $\epsilon = \frac{2\pi}{|\Theta|}$. For each of the unknown angles $\theta_i$, the server will randomly guess an index $n_i \in \{0,1,\dots,|\Theta|\}$ and select a random angle $n_i \epsilon \leq \theta_i^{\sim} \leq (n_i + 1) \epsilon$. If all their guesses are correct, they can successfully execute the computation, since $|\theta_i - \theta_i^{\sim}| \leq \epsilon$, and for sufficiently small $\epsilon$, the statistics of measurement outcomes corresponding to these two angles will be similar. Thus, in this case, the dimension of the parameter space must grow similarly to prevent the random guess attack from succeeding.
\end{enumerate}
Our lower bound is also consistent with other types of blind quantum computing protocols that utilize resources other than $|+_{\theta}\rangle$. For example, the authors in \cite{giovannetti2013efficient} proposed an efficient blind quantum computing protocol with a communication complexity of $O(J \log n)$, where $n$ is the number of qubits needed for the computation (i.e., the size of the graph state in our case), and $J$ is the computational depth. This upper bound is also consistent with our result and can be derived from the no-programming theorem \cite{nielsen1997programmable}.\\

\section{Possibility Results: Selective blindness}
\label{sec:selectively-blind-quantum-computing}

In earlier sections, the client's goal was to delegate an arbitrary computation to the server, with the server only learning an upper-bound on the size of the computation. In this section we assume that the client wishes to delegate one of two known computations. The server should not be able to guess with probability higher than $1/2$ which of these two computations the client decided to perform. This covers use-cases in which the Client is willing to leak some information about the computation they are performing.\footnote{This is the case for instance if the server knows already the ansatz that the client wants to use is a parametrized quantum circuit (for some QML application) but the client does not want to leak the parameters of the trained model.} We show that restricting our security guarantees in this way allows us to decrease drastically the qubit communication complexity of the protocol.

We first describe how to mask all possible differences between two unitary transformations expressed as MBQC measurement patterns and then give the protocol which achieves this new masking goal. Our protocol can be generalized in a straightforward way to allow the client to delegate one computation within a family of $n$ unitaries without revealing which one has been performed. This is done by considering the set of all pairwise differences between the measurement patterns in the set and masking them accordingly.

In the language of Abstract Cryptography, we design a protocol realizing the ideal functionality of \emph{Selectively Blind Quantum Computation} described in Figure~\ref{fig:sdqc}.
\begin{figure}[ht]
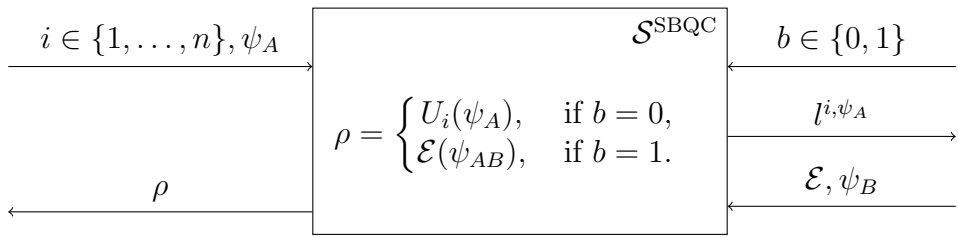

\centering
    \begin{bbrenv}{resource}
    \begin{bbrbox}[name=$\mathcal{S}^\text{SBQC}$, minheight=3cm]
        $\rho = \left\{ \begin{matrix}
             U_i(\psi_A), & \text{ if } b=0, \\
            \mathcal{E}(\psi_{AB}), & \text{ if } b=1.
        \end{matrix} \right.$
    \end{bbrbox}
    \bbrmsgspace{6mm}
    \bbrmsgto{top={$i\in\{1,\dots,n\}, \psi_A$}, length=4cm}
    \bbrmsgspace{12mm}
    \bbrmsgfrom{top={$\rho$}, length=4cm}
    \bbrqryspace{6mm}
    \bbrqryfrom{top={$b\in\bin$}, length=3cm}
    \bbrqryspace{2mm}
    \bbrqryto{top={$l^{i,\psi_A}$}, length=3cm}
    \bbrqryspace{2mm}
    \bbrqryfrom{top={$\mathcal{E},\psi_B$}, length=3cm}
    \end{bbrenv}
\caption{The ideal functionality of \emph{Selectively Blind Quantum Computation}. Note, that the selection of unitaries $\{U_i\}_{i=1,\dots,n}$ is fixed as a public parameter of the functionality, which is available to all participating parties. An honestly participating server does not have access to its interfaces, and its input $b$ is filtered to be $0$.}
\label{fig:sdqc}
\end{figure}

\subsection{Angle Masking Techniques}\label{sec:angle_masking}

In this section we describe how to hide a difference in a default measurement angle of a node in the computation. We assume that the client wishes to delegate either $U_0$ or $U_1$ on input $|0^m\ket$ for input register size $m$, whose measurement patterns when expressed in MBQC differ only in the default measurement angle at node $v$: it is $\varphi_0$ for $U_0$ and $\varphi_1$ for $U_1$. They have the same resource graph noted $G$ and the same g-flow $g$. We assume that the server knows the descriptions of $U_0$ and  $U_1$ including $\varphi_0, \varphi_1$.

\subsubsection{Target Node}

We call $v$ a \emph{target node}. 
To mask the associated angle we apply the technique of the original UBQC described in the Section \ref{subsec:ubqc}: the client sends a qubit $|+_{\theta_v}\ket$ and during the computation it sends the associated angle $\delta_v$. This perfectly hides the measurement angle for this node. 

However, this is not sufficient to hide the outcome of the computation. We must also mask all the qubits influenced by the measurement on qubit $v$, otherwise the server may use the difference in outputs to distinguish $U_0$ and $U_1$. As explained in Section \ref{subsec:mbqc} the measurement outcome $s_v$ propagates through the graph in the form of measurement corrections. For example, the target node induces an X-correction on the qubits at positions $w \in g(v)$. If this is the only correction for one such node $w$, then $w$ has to be measured with the updated angle $(-1)^{s_v} \varphi_{w}$. The Server is aware of the default measurement angle for node $w$ and they need to know the updated angle to perform the computation, hence they can calculate $s_v$.\footnote{If $w$ receives corrections from other nodes, the other correction bits are known to the server and they can again recover $s_v$ with a simple XOR.} This means that the measurement angles that depend on $s_v$ have to be masked.

As seen in Figure \ref{fig:gflow}, the target node propagates different corrections to different qubits. We analyze which masking techniques are required depending on the type of correction and the value of the measurement angle of a qubit.

\subsubsection{X Corrections} 

Consider $w \in g(v)$. The corrected angle for it is defined in the Equation \ref{eq:corrections}. Since $v$ belongs to $s_X$ of $w$, the value of $s_v$ has the effect of changing the sign of $\varphi_w$. We consider the following three cases depending on the value of $\varphi_w$.

\paragraph{No effect.} If $\varphi_{w} \in \{0, \pi\} $, the value $(-1)^{s_{v}} \varphi_{w}$ does not change depending on $s_{v}$ and we do not need to mask it.

\paragraph{Classical hiding.} Suppose now that $\varphi_{w} \in \{\frac{\pi}{2}, \frac{3\pi}{2}\} = \{ \frac{\pi}{2} + r_1 \pi | r_1 \in \{0,1\} \}$. The X-correction maps these angle to one another. To mask this change it is enough to sample $r_2 \sample \{0,1\}$ and ask the qubit to be measured with angle $\varphi''_w \coloneqq \varphi'_{w} + r_2 \pi$. Then we have
\[
    \varphi''_w =  \frac{\pi}{2} + (r_1 +r_2) \pi  + s_Z\pi
\]
It is easy to see that $r_1$ is perfectly hidden. Also, we can deduce the result of the original measurement by flipping the outcome transmitted by the Server when $r_2 = 1$ or keeping it as it is when $r_2 = 0$.

\paragraph{Quantum hiding.} Lastly, assume $\varphi_{w} \in \{\frac{(2k+1)\pi}{4}| k=0, \dots, 3 \} $. For these nodes the X-correction changes the angle from $\pi/4$ to $7\pi/4$ or from $3\pi/4$ to $5\pi/4$. It is not possible to deterministically deduce the result of a measurement in the $\pi/4$ basis from a measurement outcome in the $7\pi/4$ basis and vice versa, and similarly for the other two bases, since they differ by a value of $\pi/2$. Hence, to mask them the client sends a $| +_{\frac{r_2\pi}{2}} \ket$ qubit to the Server, where $r_2 \sample \{0,\dots,3\}$ and change the measurement angle to $\varphi''_{w} = \varphi'_{w} + \frac{r_2\pi}{2} + r_3 \pi$. Here $r_3 \sample \{0,1\}$ is hiding the measurement outcome. The resulting measurement angle is uniformly random among the four possible choices, so it is perfectly hidden. This last case therefore incurs a communication cost of $1$ qubit.

\begin{remark}
    Note that X- and Z-corrections do not change the angle type. That is why instead of talking about the type of the corrected angle $(-1)^{s_X} \varphi_w + s_Z\pi$ we just consider the default value $\varphi_w$.
\end{remark}

\subsubsection{Z Corrections}

We now consider nodes $u$ where $u\in Odd(g(v))$. If $s_v$ is equal to $1$ it has an effect of adding $\pi$ to the measurement angle of $u$. To achieve indistinguishability, the Client transmits a corrected measurement angle $\varphi''_{u} = \varphi'_u + r\pi$ with $r \sample \{0,1\}$. Since $r$ is not known to the adversary, the value of $s_{v}$ in $\varphi'_u$ is perfectly hidden. The real measurement result can be deduced similarly to the previous cases.

\subsubsection{Propagating Corrections}

Similarly to the corrections coming directly from the target node, we need to mask the ones produced by the dependent nodes of the target. We define the \emph{future cone} as the set of vertices that need to be masked. Intuitively, the future cone of a node $v$ contains all the qubits whose measuring angle depends on the measurement outcome of $v$. See Figure \ref{fig:corrections} for an example of a future cone.

\begin{definition}[Future cone]
\label{def:future_cone}
The \emph{future cone} of a node $v$ in a resource graph for MBQC computation is defined recursively. A qubit $w$ belongs to the future cone of $v$ if
\begin{itemize}
    \item It belongs to the set of X or Z dependencies of $v$.
    \item It belongs to the set of X or Z dependencies of a node in the future cone.
\end{itemize}
The set of qubits that receive an X-correction from the target node or from the other nodes of the future cone are called the \emph{interior} of the future cone.
\end{definition}

\begin{definition}[Qubit-masked node]
A node $w$ is called \emph{qubit-masked} if it belongs to the interior of the future cone of the target qubit and its default measurement angle belongs to the set $\{\frac{\pi}{4},\frac{3\pi}{4},\frac{5\pi}{4},\frac{7\pi}{4}\}$.
\end{definition}

\begin{figure}[ht]
\centering
\begin{tikzpicture}
  \def\rows{6}
  \def\cols{6}
  \foreach \x in {1,...,\cols} {
    \foreach \y in {1,...,\rows} {
      \def\fillcolor{white}
      \pgfmathtruncatemacro{\dx}{\x - 2}
      \pgfmathtruncatemacro{\dy}{\y - 3}
      \ifnum\x=2
        \ifnum\y=3
          \def\fillcolor{red!80!black}
        \fi
      \else
        \ifnum\dx>0
          \ifnum\dy=\dx
            \def\fillcolor{blue!80!black}
          \else
            \ifnum\dy=-\dx
              \def\fillcolor{blue!80!black}
            \else
              \ifnum\dy<\dx
                \ifnum\dy>-\dx
                  \def\fillcolor{green!50!black}
                \fi
              \fi
            \fi
          \fi
        \fi
      \fi
      \node[draw, circle, fill=\fillcolor, minimum size=0.4cm, inner sep=0pt] (N\x\y) at (\x, \y) {};
    }
  }
  \foreach \x in {1,...,\cols} {
    \foreach \y in {1,...,\rows} {
      \pgfmathtruncatemacro{\xp}{\x + 1}
      \pgfmathtruncatemacro{\yp}{\y + 1}
      \ifnum\x<\cols
        \draw (N\x\y) -- (N\xp\y);
      \fi
      \ifnum\y<\rows
        \draw (N\x\y) -- (N\x\yp);
      \fi
    }
  }
\end{tikzpicture}
\caption{The future cone of the red node is shown in green and blue. The green nodes denote the interior of the future cone. The term is inspired by the shape of this dependence set in the resource graph equal to the $\mathbb{Z}_2$ lattice with the g-flow $g(x,y) = \{(x,y+1)\}$.}
\label{fig:corrections}
\end{figure}
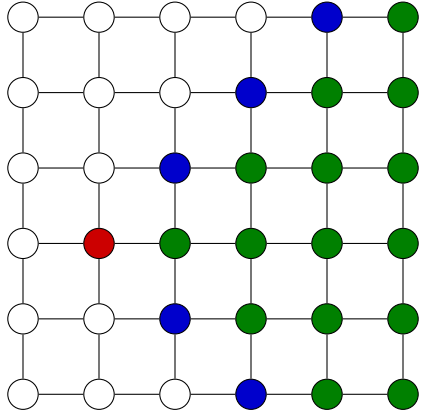

\begin{recap}
\vspace*{1ex}
\caption{Angle hiding procedure.}
\begin{itemize}
    \item A target node generates a future cone.
    \item The target node and qubit-masked nodes are masked by having the client send additional qubits and adapting the measurement angle.
    \item The client hides the other nodes of the cone classically by adding a random $r\pi$ rotation to the measurement angle.
\end{itemize}
\end{recap}

\begin{remark}[Generalization to private quantum input.]\label{sec:private_input}
To hide a quantum input we use the quantum one time pad $Z(\theta_i)X^{a_i}(\rho_{in})$. For every index $i$ we can see the part $a_i$ of the secret key as an X-correction to an input qubit. Since these qubits receive this X-correction and some of them might be qubit-masked we hide the effect of this correction using $\theta_i$. That is why instead of a simple $Z^{b_i}$ we use $Z(\theta_i)$ in the encryption.

We need to hide the measurement outcome of these qubits. We can see this scenario as an instance of angle hiding where all input nodes are target nodes. The quantum communication cost in this case is equal to
\begin{align*}
    \#\{\text{input/output qubits}\} + \# \left\{ \frac{(2k+1)\pi}{4} \text{ angles on dependent qubits} \right\}.
\end{align*}
\end{remark}

\begin{remark}[Further optimizations.]
Note that we can optimize the angle hiding algorithm even further for the qubit that is sent for the target node. Consider the case when $\varphi_0 - \varphi_1 = \pi \pmod{2\pi}$ and both angles $\varphi_0$ and $\varphi_1$ are Clifford angles. Then the corrections can only change these angles by $\pi$. Therefore, this difference can be concealed using classical masking and we do not need to send a qubit for this node. Nevertheless, we will have to mask the future cone of this node.

Additionally, for most of the other cases it is enough to send a qubit of the form $|+_{\theta}\ket$ where $\theta \in \frac{k\pi}{2}, k = 0, \dots, 3$. The only time when all 8 possibilities are required is when the difference $\varphi_0 - \varphi_1 = \frac{(2k+1)\pi}{4} \pmod{2\pi}$, with $k=0, \dots, 3$.
\end{remark}

\subsection{Graph Masking Techniques}\label{sec:graph_masking}

Suppose now that the two unitaries $U_0 \coloneqq (G_0, g_0, (\varphi^{0}_{i})_{i=0}^{n_0})$ and $U_1 \coloneqq (G_1, g_1, (\varphi^{1}_{i})_{i=0}^{n_1})$ that have different resource graphs, g-flow functions, input and output spaces, but all default measurement angles are equal to $0$: $\forall i \forall b: \varphi^{b}_{i} = 0$. We add this constraint to only deal with the difference in the graphs and g-flows as we have already discussed how to mask the default angles in Section \ref{sec:angle_masking}.

We first give a necessary and sufficient condition for merging $G_0$ and $G_1$ into one graph. Then we describe a number of graph masking techniques that can be applied to the merger of $G_0$ and $G_1$ to make the computations defined on these graphs indistinguishable. We leave the question of finding the optimal merger graph for future work.

\subsubsection{Condition for Merging Graphs.}

Suppose we are given a pair of resource graphs with g-flows on them $(G_0, g_0, \prec_0)$ and  $(G_1, g_1, \prec_1)$. We define a third graph $(G_M, \prec_{G_M})$ such that it ``contains'' both $G_0$ and $G_1$ and has an order of measurements that does not violate the orderings on $G_0$ and on $G_1$.

Having $G_M$ we can use the techniques described below to reduce $G_M$ to $G_0$ or $G_1$ following the choice of the client, while keeping the measurement order the same for both computations.

\begin{definition}[Merger graph]
\label{def:merger_graph}
We define a graph $G_M$ to be the merger of two given graphs $G_0$ and $G_1$ with their respective partial orders $\prec_0$ and $\prec_1$ if it has the following properties
\begin{enumerate}
    \item \label{merging:emb1} $\exists$ an embedding function $i_0: G_0 \rightarrow G_M$ preserving edges.
    \item \label{merging:emb2} $\exists$ an embedding function $i_1: G_1 \rightarrow G_M$ preserving edges.
    \item \label{merging:order} $\exists$ a total order $\prec_{G_M}$ such that, if $a \prec_0 b$, for some $a,b \in G_0$ than $i_0(a) \prec_{G_M} i_0(b)$ and similarly for $G_1$.
\end{enumerate}
\end{definition}
\noindent In the rest of the document we will abuse notation and say that some of the nodes of $G_M$ belong to $G_0$ or $G_1$ instead of saying that they are an image of a node from $G_0$ or $G_1$ under an embedding function.

A trivial instantiation of $G_M$ is a graph that contains $G_0$ and $G_1$ as two non intersecting components. Then one can measure all the non-output vertices of $G_0$ first and followed by all of the non-output vertices of $G_1$. This approach is obviously not the most efficient. In this section we set the necessary and sufficient conditions on $G_M$ to be a merger of $G_0$ and $G_1$.

\begin{lemma}
\label{lem:merger_graph}
    A graph $G_M$ is a merger of $G_0$ and $G_1$ if and only if the transitive closure of $\prec_0 \cup  \prec_1$ (after the respective renaming of vertices) is a strict partial order on $G_M$ i.e. it is irreflexive, antisymmetric and transitive.
\end{lemma}
\begin{proof}
    Firstly, note that irreflexivity of the union of two partial orders is obvious, so we are not going to discuss it in the proof.

    We assume $G_M$ is a merger of $G_0$ and $G_1$ and prove that the transitive closure of $\prec_u \coloneqq \prec_0 \cup  \prec_1$ is a partial order. By the definition, two arbitrary elements $a,b \in G_M$ have $a \prec_u b$ if and only if $i_0^{-1}(a) \prec_0 i_0^{-1}(b)$ or $i_1^{-1}(a) \prec_1 i_1^{-1}(b)$.

    Take the total order $\prec_{G_M}$ from part \ref{merging:order} of the merger graph definition. From the statement above, we can see that if $a \prec_u b$ then $a \prec_{G_M} b$. Therefore, $\prec_u$ is a suborder of $\prec_{G_M}$ hence it cannot have relations that violate transitivity or antisymmetry.

    Nevertheless, it is possible that we have $a \prec_u b$ and $b \prec_u c$ but nothing for $a$ and $c$. It can happen if $b$ has a preimage in $G_0$ and in $G_1$ and the first relation comes from $\prec_0$ but the second from $\prec_1$. The relations that complete transitivity can be safely added to $\prec_u$ since they are present in $\prec_{G_M}$. In conclusion, we have obtained that a transitive closure of $\prec_u$ is a strict partial order.

    Suppose now that we have some embeddings for $G_0$ and $G_1$ into $G_M$ and $\prec_u \coloneqq \prec_0 \cup \prec_1$ is a strict partial order. We would like to prove that there exists a total order compatible with $\prec_0$ and $\prec_1$. Let us take $\prec_{G_M}$ to be a linear extension of $\prec_u$ (it exists by the order-extension principle) and check if property \ref{merging:order} holds. For an arbitrary pair $a,b \in G_0$ if we have $a \prec_0 b$ then $i_0(a) \prec_u i_0(b)$ and therefore $i_0(a) \prec_{G_M} i_0(b)$. Similarly for $a,b \in G_1$ if we have $a \prec_1 b$ then $i_1(a) \prec_u i_1(b)$ and so $i_1(a) \prec_{G_M} i_1(b)$. Hence $G_M$ is a merger of $G_0$ and $G_1$.

\end{proof}

\begin{remark}
    Note that the ordering on the merger does not depend on g-flow functions $g_0$ and $g_1$. These functions guarantee that $\prec_0$ and $\prec_1$ satisfy the properties needed for the computation to be deterministic. Therefore every extension of $\prec_i$ to a total order will be compatible with computation $i$. It is thus sufficient to find a total order that is an extension of both $\prec_0$ and $\prec_1$.
\end{remark}

\subsubsection{Masking the Edges.}
\label{subsec:shape diff}

In this section we only consider the case where $G_0$, $G_1$ and $G_M$ have the same number of vertices and the same input and output nodes, so after they are merged we only need to delete edges from the merger graph to obtain either $G_0$ or $G_1$. The technique we use in this section utilizes the bridge and break protocol, firstly developed by \cite{MPKK18}.
Suppose we have three nodes of the resource graph prepared as in Figure~\ref{subfig:bridge-break-input}.

The bridge operation deletes the vertex in the middle and creates an edge between the side vertices. The break operation also deletes the middle vertex but leaves the side vertices disconnected.

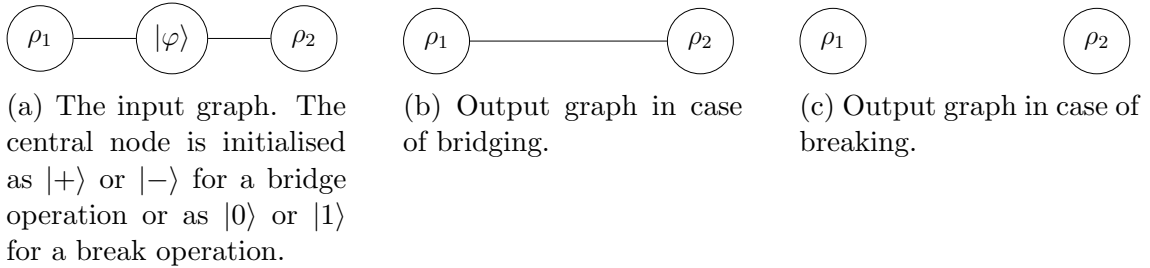
\begin{figure}[ht]
\centering
\begin{subfigure}[t]{0.3\textwidth}
\resizebox{\textwidth}{!}{%
\begin{tikzpicture}
  \node[draw, circle, minimum size=1cm] (v1) at (0,0) {$\rho_1$};
  \node[draw, circle, minimum size=1cm] (v2) at (2,0) {$|\varphi\rangle$};
  \node[draw, circle, minimum size=1cm] (v3) at (4,0) {$\rho_2$};
  \draw (v1) -- (v2);
  \draw (v2) -- (v3);
\end{tikzpicture}
}
\caption{The input graph. The central node is initialised as $|+\rangle$ or $|-\rangle$ for a bridge operation or as $|0\rangle$ or $|1\rangle$ for a break operation.}
\label{subfig:bridge-break-input}
\end{subfigure}
\hfill
\begin{subfigure}[t]{0.3\textwidth}
\resizebox{\textwidth}{!}{%
\begin{tikzpicture}
  \node[draw, circle, minimum size=1cm] (v1) at (0,0) {$\rho_1$};
  \node[draw, circle, minimum size=1cm] (v3) at (4,0) {$\rho_2$};
  \draw (v1) -- (v3);
\end{tikzpicture}
}
\caption{Output graph in case of bridging.}
\label{subfig:bridge}
\end{subfigure}
\hfill
\begin{subfigure}[t]{0.3\textwidth}
\resizebox{\textwidth}{!}{%
\begin{tikzpicture}
  \node[draw, circle, minimum size=1cm] (v1) at (0,0) {$\rho_1$};
  \node[draw, circle, minimum size=1cm] (v3) at (4,0) {$\rho_2$};
\end{tikzpicture}
}
\caption{Output graph in case of breaking.}
\label{subfig:break}
\end{subfigure}
\caption{Bridge and break states before and after the protocol.}
\label{fig:bridge and break}
\end{figure}

Let us give a detailed description of both protocols in the following diagrams.

\begin{protocol}
\vspace*{1ex}
\caption{Bridge operation.}\label{prot:bridge}
\begin{algorithmic}[0]
\STATE \textbf{Inputs:} The state $CZ_{1,2}CZ_{2,3}(\rho_1|\varphi \ket \rho_2)$ which is equivalent to the Figure \ref{fig:bridge and break}A. Here $|\varphi \ket = \frac{1}{\sqrt{2}}(|0\ket + (-1)^{c}|1\ket)$, with $c \in \{0,1\}$
\STATE \textbf{Outputs:} The state $CZ_{1,2}(\rho_1 \rho_2)$ (as on the Figure \ref{fig:bridge and break}B).
\STATE  \textbf{Protocol:}
\begin{enumerate}
    \item \label{bridge:step1} Measure the second qubit in $|+_{\pi/2} \ket, |-_{\pi/2} \ket$ basis. Record the result in $b$.
    \item \label{bridge:step2} Apply a rotation $Z(-\frac{\pi}{2})$ to the side qubits.
    \item \label{bridge:step3} Apply a correction $Z^{b+c}$ to the side qubits.
\end{enumerate}
\end{algorithmic}
\end{protocol}

\begin{protocol}
\vspace*{1ex}
\caption{Break operation.}
\label{prot:break}
\begin{algorithmic}[0]

\STATE \textbf{Inputs:} The state $CZ_{1,2}CZ_{2,3}(\rho_1 |\varphi \ket \rho_2 )$. Here $|\varphi \ket = |c\ket$, with $c \in \{0,1\}$

\STATE \textbf{Outputs:} Qubits $\rho_1 \rho_2$ disentangled (as on the Figure \ref{fig:bridge and break}C).

\STATE  \textbf{Protocol:}

\begin{enumerate}
    \item \label{break:step1} Measure the second qubit in an arbitrary basis.
    \item \label{break:step2} Apply a correction $Z^{c}$ to both side qubits.
\end{enumerate}
\end{algorithmic}
\end{protocol}

Given the description of bridge and break operations, The idea behind our edge masking technique is to add a node in the middle of each edge that originally belonged only to one of the graphs $G_0$ or $G_1$, called \emph{middle nodes}. Then the server will transform $G_M$ into $G_0$ or $G_1$ using bridge and break operations (Protocol \ref{prot:bridge} and Protocol \ref{prot:break}) in a blind way. Both of these operations consist of a measurement of the middle node and rotations to the side vertices linked to this middle node. Our goal is to incorporate these procedures into a single protocol so that they are indistinguishable for the server.

This is done by first having the client send a qubit in a state chosen at random from $\{|0\ket, |1\ket\}$ if they want to break the edge, or from $\{|+\ket, |-\ket\}$ if they want to bridge it. The server is then instructed in both cases to measure all middle qubits in the basis $|+_{\frac{\pi}{2}}\ket, |-_{\frac{\pi}{2}}\ket$ before any other operation for the computation.

The only other difference between the Protocol \ref{prot:bridge} and Protocol \ref{prot:break} consists of $Z(-\frac{\pi}{2})$ rotations on the side vertices and Z-corrections. We incorporate these operations into the measurement of the corresponding nodes. To hide whether a $(-\frac{\pi}{2})$-rotation is performed, the client must send a qubit in a state $|+_\theta \ket$ with $\theta \sample \{\frac{k\pi}{2}, k = 0, \dots, 3\}$ for both of the side vertices. Later in the computation phase of the protocol we measure this node in the basis $|+_\delta \ket, |-_\delta \ket$ where 
\begin{align*}
    &\delta = \psi' + \theta - \frac{\pi}{2} + (b+c)\pi + r\pi \text{ for the bridging, } \\
    &\delta = \psi' + \theta + c\pi + r\pi \text{ for the breaking }
\end{align*} with $r \sample \{0,1\}$ and following the notations of Protocols \ref{prot:bridge} and \ref{prot:break}. Then $\theta$ and $r$ fully hide the difference between bridging and breaking and the server is incapable of distinguishing the two operation based on the qubits received from the client.

Similarly to the angle hiding, we need to mask the measurement results of these nodes, so we apply our hiding methods to their future cones. The qubit cost of this operation is equal to 1 qubit for the middle vertex plus 2 qubits for the side vertices, plus the number of qubit-masked nodes in their future cones.

\begin{recap}
\vspace*{1ex}
\caption{Masking the shape of the graph.}
\begin{itemize}
    \item The middle vertex is in one of the states: $|0\ket, |1\ket, |+\ket, |-\ket$ depending on the operation.
    \item It is measured in $|+_{\frac{\pi}{2}}\ket, |-_{\frac{\pi}{2}}\ket$ basis.
    \item Side vertices need a qubit to mask possible corrections and generate future cones.
\end{itemize}
\end{recap}

\subsubsection{Deleting a Computational Vertex.}
\label{subsubsec:delete comp}
Here we describe a procedure for blindly deleting a computation vertex $v \in G_M$ that has a preimage under the embedding function in $G_0$ but no preimage in $G_1$. Input and output vertices are discussed respectively in Section \ref{subsubsec:input} and Section \ref{subsubsec:output}. 

If the client wants to execute $U_1$ it deletes $v$ from $G_M$ by sending a state $| b \ket \in \{|0\ket, |1\ket \}$ for this node which disentangles $v$ from the rest of the graph. Every neighboring state it corrected by applying $Z^{b}$, i.e. adding $b\pi$ to their measurement angle. The node $v$ is measured by the server in the basis $\delta = \psi' + \theta$ for a randomly chosen $\theta \sample \{\frac{k\pi}{2} : k = 0 \dots 3\}$.

If the client want to execute $U_0$,  it sends a qubit for vertex $v$ in a states chosen at random from $\{|+_{\frac{k\pi}{2}}\ket : k = 0 \dots 3\}$\footnote{
    Note that in this situation the Client never sends one of 8 angles for this node as it only happens when there is an angle difference between $U_0$ and $U_1$.}
and adapts the measurement angle to $\delta = \psi' + \theta + r\pi$, where $\theta$ is the angle chosen for the state and $r$ is a random bit. 

In both cases, the client will have to hide the measurement outcome of this node and the neighbors so they have to apply the hiding techniques to the future cones of these qubits.

\begin{recap}
\vspace*{1ex}
\caption{Masking a deleted vertex.}
\begin{itemize}
    \item The Client sends a $|0\ket, |1\ket$ state to delete a vertex and $|+_\theta \ket$ with $\theta \in \{\frac{k\pi}{2}, k = 0, \dots, 3 \}$ to keep it.
    \item The node is measured with angle $\delta = \psi' + \theta + r\pi$ and generates a future cone.
    \item Its neighbours are masked with an $r\pi$ rotation and generate future cones.
\end{itemize}
\end{recap}

\subsubsection{Quantum Input Vertices.}
\label{subsubsec:input}

We do not require the embeddings of $G_0$ and $G_1$ to have the same position and number of the input vertices. To make the two computations indistinguishable the client sends qubits for every vertex of $G_M$ such that at least one of its preimages under the embedding functions is an input vertex. Let $v$ be an input node of $G_0$.

If the initial state $\rho_{in}$ for this input vertex is known to the server (public input), its vertex in $G_M$ is mapped to a public input in $G_1$ and they have the same known input state distribution then no masking is required. In other cases if the client whishes to run $U_0$ -- including if the input is private to the client -- the client encrypts the state it sends  $v$ with $Z(\theta)X^a$ where $\theta \sample \{\frac{k\pi}{2}, k = 0, \dots, 3\}$ and $a \sample \{0,1\}$.

If the client wishes to run $U_1$ instead, the behavior depends on the type of the node. If $v$ is also an input of $G_1$ but they do not have the same distribution, then the state sent by the client for vertex $v$ in $U_1$ should be encrypted in the same way. If $v$ is a deleted node in $G_1$, the client applies the procedure from Section \ref{subsubsec:delete comp}.

Otherwise, $v$ is a computational node of $G_1$. If these nodes have a difference in their measurement angles, the client applies the procedure from the Section \ref{sec:angle_masking}.
If not, the client sends a qubit $|+_{\theta} \ket$, with $\theta \sample \{\frac{k\pi}{2}, k = 0, \dots, 3\}$ and adapt the measurement angle accordingly.

\begin{recap}
\caption{Masking input vertices.}
\vspace*{1ex}
\textit{Case 1.} When a public input node collides with another public input node with the same distribution then it is transmitted in clear and does not generate a cone.

\textit{Case 2.} When an input node of $U_0$ is either private, or public and:
\begin{itemize}
    \item[a.] collides with a computational node of $U_1$.
    \item[b.] collides with an input node of $U_1$ with a different distribution.
    \item[c.] collides with an encrypted input.
    \item[d.] collides with a deleted node of $U_1$.
\end{itemize}
\textit{Actions:}
\begin{enumerate}
    \item Send $X^aZ(\theta)(\rho_{in})$ for $U_0$.
    \item For the same order of cases for $U_1$ we send a qubit:
    \begin{itemize}
        \item[a.] Send $|+_\theta \ket$ with $\theta \sample \{\frac{k\pi}{2}, k = 0, \dots, 3\}$.
        \item[b.] Send $X^aZ(\theta)(\rho_{in})$ with $\theta \sample \{\frac{k\pi}{2}, k = 0, \dots, 3\}$.
        \item[c.] Send $X^aZ(\theta)(\rho_{in})$ with $\theta \sample \{\frac{k\pi}{2}, k = 0, \dots, 3\}$.
        \item[d.] Send a $|b\ket$ state with $b \sample \{0,1\}$.
    \end{itemize}
    \item Measure with angle $\delta = (-1)^{a}\psi + \theta + r\pi$.
    \item Mask the future cone generated by the node.
\end{enumerate}
\end{recap}

\subsubsection{Quantum Output Vertices.}
\label{subsubsec:output}

The server should not be able to distinguish the two unitaries by the location or value of the output. If our computation has classical output, then the problem is trivial since the operation of obtaining the output is equivalent to measuring a computational vertex and the future cone masking technique hides the outputs of exactly all vertices that may be influenced by the difference between the two computations. One the other hand, if the output is quantum then the output qubits have to be sent to the client. The set of qubits the server returns to the client should be the same for both computations.

When the locations of the output nodes do not coincide we will teleport the output qubits of the ``smaller graph'' into the output set of $G_M$. Let $v$ be an output vertex of $G_0$ and a computational vertex of $G_1$. We add to $G_M$ two vertices $v_1, v_1$ and the edges $(v, v_1)$ and $(v_1, v_2)$ as shown in Figure \ref{fig:output}. These additional vertices are considered as part of $G_0$. The new default measurement angles are $0$ for $v$ and $v_1$. This means that the operation applied is the identity. To mask this procedure, the client only needs to delete all vertices that do not exist for the computation being executed as described in Section~\ref{subsubsec:delete comp}. The Server prepares the qubit for $v_2$ in the state $|+\ket$. The qubit for $v_2$ is sent back to the client at the end of the computation as it contains the quantum output if the computation is $U_0$.

\begin{figure}[ht]
\caption{Teleporting the output values. Red vertices (including dotted red) belong to $G_1$, yellow vertices belong to $G_0$. The newly created vertices and edges are shown in green and the output of $G_M$ are the nodes in zigzag blue. }
\centering
\includegraphics[scale=0.4]{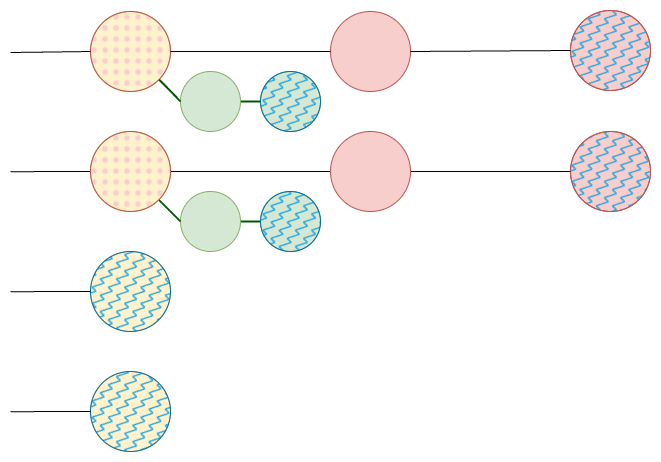}
\label{fig:output}
\end{figure}


\begin{recap}
\vspace*{1ex}
\caption{Masking output vertices.}
Consider a node $v$ that is used as a computational vertex in one of the unitaries and as the output vertex in the other. Let $v_1$ and $v_2$ be the added vertices. The client does the following:
\begin{itemize}
    \item When $v$ is used as a computational vertex:
    \begin{enumerate}
        \item Send $|b\ket$ with $b \sample \{0,1\}$ for $v_1$.
        \item Measure $v$ with angle $(r + b)\pi$ for $r \sample \{0, 1\}$.
        \item Measure $v_1$ with angle $r_1\pi$ for $r_1 \sample \{0, 1\}$.
        \item Mask $v$'s future cone.
    \end{enumerate}
    \item When $v$ used as output:
    \begin{enumerate}
        \item Send a state from $\{|+\ket, |-\ket\}$ for $v_1$.
        \item Measure $v$ with the default angle $r\pi$ for $r \sample \{0, 1\}$.
        \item Measure the $v_1$ with default angle $r_1\pi$, $r_1 \sample \{0,1\}$.
        \item Mask $v$'s future cone.
    \end{enumerate}
\end{itemize}
\end{recap}

\subsubsection{Masking the G-flow.}

The last aspect of the computation that has to be hidden is the g-flow. From the definition of g-flow, each node receives the following corrections:

\[
    \varphi'_i = (-1)^{s_X}\varphi_i + s_{Z}\pi
\]
where
\[
    s_X = \sum_{j: i \in g(j)} s_j\text{,  }
    s_Z = \sum_{j: i \in Odd(g(j)), j \neq i} s_j
\]

We define a X-(Z-)correction set of qubit $i$ to be the set of qubits that may give $i$ an X-(Z-)correction. Here the X-correction set of $i$ is equal to $\{j: i \in g(j), j \neq i\}$ and Z-correction set is $\{j: i \in Odd(g(j))\}$.

The masking of a node is going to depend on the difference in correction sets. On one hand, if both $U_0$ and $U_1$ produce the same correction sets for the node $i$ we apply no hiding.

If only Z-correction sets are different then $\varphi^{0}_i - \varphi^{1}_i = k\pi \pmod{2\pi}$ and classically $r\pi$-masking the node and masking the future cone of this node will suffice.

If there is a difference in X-correction sets, this is equivalent to hiding an X-correction on a node of a future cone from the Section \ref{sec:angle_masking}. Therefore, it can be hidden with an $r\pi$ rotation if its measurement angle is one of $\{\frac{k\pi}{2}, k = 0, \dots, 3\}$, or the client sends an additional qubit for it otherwise. In both cases we mask the future cone of this qubit.

\begin{remark}
Note that these influences are measurement result dependent. Suppose we are working with a node $v$ and every node in the difference between its X-correction sets is already masked with an $r\pi$ rotation. Then the attacker does not know the measurement outcome for these nodes and cannot tell where a possible X-correction of $v$ is coming from. Hence, we do not need to hide the X-correction difference for $v$. The same holds for the Z-correction sets.
\end{remark}

\begin{recap}
\caption{Masking the g-flow difference.}
\vspace*{1ex}
\begin{enumerate}
    \item Compute the correction sets of the node.
    \item If there is a difference in Z-correction sets mask the qubit with an $r\pi$ rotation.
    \item Else if there is a difference in X-correction sets mask the qubit with an $r\pi$ rotation when the angle is from $\{\frac{k\pi}{2}, k = 0, \dots, 3\}$ or using an additional qubit when it is from $\{\frac{(2k+1)\pi}{4}, k = 0, \dots, 3\}$ otherwise.
    \item In both cases mask the future cone of the qubit.
\end{enumerate}
\end{recap}

\subsection{Full Masking Protocol}\label{sec:full_protocol}

After describing in the previous two sections the various cases which require some hiding technique to be applied, we combine these cases here and characterize more generally when future cone masking is necessary. This will facilitate the description of our full protocol later.

We define influence sets of every node with respect to the operations induced by the masking techniques or by the MBQC corrections. These operations include Z-corrections, X-corrections and $\frac{\pi}{2}$-rotations.

\begin{definition}[X-influence set]
    The X-influence set of a node $v$ in graph $G$ contains all nodes of $G$ that require applying an $X$ gate to $v$ at least in one of the branches of the computation to keep the result deterministic and correct.
\end{definition}

\begin{definition}[Z-influence set]
    The Z-influence set of a node $v$ in graph $G$ contains all nodes of $G$ that require applying a $Z$ gate to $v$ at least in one of the branches of the computation to keep the result deterministic and correct.
\end{definition}

\begin{definition}[R-influence set]
    The R-influence set of a node $v$ in graph $G$ contains all the nodes of $G$ that require applying a rotation gate $Z(-\frac{\pi}{2})$ to $v$ to keep the result correct.
\end{definition}

We can now define the set of vertices in $G_M$ that generate future cones. This set includes every vertex for which an influence set in $G_0$ and in $G_1$ is different. The other two types of vertices that generate a cone are vertices having a difference in the default measurement angle in $U_0$ and $U_1$ and the input vertices that need to be encrypted.

\begin{remark}
    Since the g-flow in $G_0$ and $G_1$ can differ, the future cones of a node are going to be different in the two cases. Therefore, we need to hide the shape of the future cone of a vertex in both graphs. When we say a future cone of a vertex in this part of the document we mean the union of the future cones of this node in both graphs.
\end{remark}


\noindent We formally define the full masking procedure in Protocol \ref{prot:full masking pp} and Protocol \ref{prot:full masking cp}.

\begin{protocol}[htp]
\caption{Full Masking. Preparation Phase.}
\vspace*{1ex}
\label{prot:full masking pp}
\begin{algorithmic}[0]

    \STATE \textbf{Inputs:} A graph $G_M$ compatible with masking techniques described above. We are assuming that the output teleportation circuits are already a part of $G_M$.
    \STATE  \textbf{Protocol: (steps \ref{pp:step1}-\ref{pp:step4} are executed by the Client)}

    \begin{enumerate}
        \item \label{pp:step1} Compute the input maskings depending on the input distribution as discussed in section \ref{subsubsec:input}.
        \item \label{pp:step2} For every node in $G_M$, compute the influence sets of this node. To do this we need to consider influences of the g-flow in both graphs, the locations of bridge and break operations, deleted vertices and output teleportation circuits.
        \item \label{pp:step3} Compute the list of nodes that generate future cones. This list should include every vertex that has different influence sets in $G_0$ and $G_1$, the vertices that have an angle difference in two graphs and encrypted input vertices.
        \item \label{pp:step4} For every node $v$ send the following qubits to the Server:
        \begin{enumerate}
            \item \label{pp:step4a} $\rho_{in}$ if $v$ is a clear text input.
            \item \label{pp:step4b} $X^{a_v}Z(\theta_v)(\rho_{in})$ with $\theta_v \sample \{\frac{k\pi}{2}, k= 0, \dots, 3\}$ and $a_v \sample \{0,1\}$ if $v$ is an encrypted input.
            \item \label{pp:step4c} $|0\ket /|1\ket$ if $v$ is a deleted node or middle break vertex.
            \item \label{pp:step4d} $|+\ket / |-\ket$ if $v$ is a middle bridge vertex or a middle teleportation qubit when output is being teleported.
            \item \label{pp:step4e} $| +_{\theta_v} \ket$ where $\theta_v$ is defined by $\theta_v \sample \{\frac{k\pi}{2}, k= 0, \dots, 3\}$ if $v$ is a computational vertex of the current graph that would be deleted or would be an input in the other graph. Also if $v$ has an X-influence difference and has a qubit-masked measurement angle or it is a qubit-masked node that belongs to a future cone.
            \item \label{pp:step4f} $| +_{\theta_v} \ket$ where $\theta_v$ is defined by $\theta_v \sample \{\frac{k\pi}{4}, k= 0, \dots, 7\}$ if $v$ has different default measurement angles in $U_0$ and $U_1$.
        \end{enumerate}
        \item \label{pp:step5} The Server prepares other nodes of the graph in a $|+\ket$ state.
        \item \label{pp:step6} The Server entangles the qubits received from the Client and the $|+\ket$ states with respect to the resource graph $G_M$.
    \end{enumerate}
\end{algorithmic}
\end{protocol}

\begin{protocol}[htp]
\caption{Full Masking. Computation Phase.}
\label{prot:full masking cp}
\begin{algorithmic}[0]

    \STATE \begin{enumerate}
        \item \label{cp:step1} The Server measures every middle bridge/break vertex with angle $\frac{\pi}{2}$.
        \item \label{cp:step1a} The Server measures every middle teleportation node with angle $r\pi$, with $r \sample \{0,1\}$.
        \item \label{cp:step2} For every remaining non-output node for $G_M$, let us call it $v$, repeat the following:
        \begin{enumerate}
            \item \label{cp:step2a} The Client computes the corrected measurement angle $\varphi'$ depending on the default angle and the real measurement outcomes of previous vertices.
            \item \label{cp:step2b} The Client instructs the Server to measure $v$ with angle $\delta$ where:
            \begin{enumerate}
                    \item \label{cp:step2b2} $\delta = \varphi' + \theta_v +r\pi$ if $v$ is an output and computational node at the same time, or has a default angle difference, or has R-influence set difference, or it is in a future cone, or it has X-influence difference and is qubit-masked; or if $v$ is a node that is deleted in one of the graphs; or if $v$ is an input and computational node at the same time; or if $v$ is an encrypted input.
                    \item \label{cp:step2b3} $\delta  = \varphi'+ r\pi$ if $v$ is a part of a future cone or has an X- or Z- influence set difference and is not qubit-masked.
                    \item \label{cp:step2b4} $\delta = \varphi'$ if there is no masking on $v$; or if there is $v$ is a public input vertex.
                \end{enumerate}
            In case of ambiguity the Client should apply the first rule from the above list that fits. The parameter $r$ is sampled uniformly at random from $\{0,1\}$.
            \item \label{cp:step2c} The Server returns the measurement outcome to the Client.
    \end{enumerate}
        \item \label{cp:step3} When every non-output qubit has been measured the Server sends the output qubits to the Client.
        \item \label{cp:step4} The Client calculates the output X- and Z-corrections that depend on the masking and the measurement outcomes and applies them to the output qubits.
    \end{enumerate}
\end{algorithmic}
\end{protocol}

\subsection{Proof of Blindness in AC Framework.}

In the following section we prove the blindness of this protocol. We start by giving a high-level definition of the Ideal Resource for the 1-of-2 DQC. Then we formulate the Security Theorem \ref{thm:blind}. To prove it we modify the original Protocols \ref{prot:full masking pp} and \ref{prot:full masking cp} in a series of reductions that are statistically equivalent. Lastly, we complete the definition describing the 1-of-2 DQC Resource and we complete the proof by designing the Simulator that makes the real world protocol statistically indistinguishable from the ideal world simulation. This implies the blindness of the scheme in the AC framework.

\subsubsection{1-of-2 Delegated Quantum Computation Resource Definition.}

Our protocol implements the 1-of-2 Delegated Quantum Computation Resource ($\mathcal{R}$) for an honest Client, a possibly malicious Server. The Server has two filtered interfaces that are controlled by bits $(c,d)$. We define the 1-of-2 DQC Resource for two fixed unitary transformations $U_0$ and $U_1$ that are embedded in $\mathcal{R}$.

\begin{resource}
\caption{1-of-2 Delegated Quantum Computation Resource.}
\label{res:ubqc short}
\begin{algorithmic}
\STATE \begin{enumerate}
    \item Client sends $i \in \{0,1\}$ and $\rho_{C}$ to the $\mathcal{R}$. The bit $i$ is a control bit to tell $\mathcal{R}$ which unitary the Client wants to execute and $\rho_{C}$ is an n-qubit long input to the requested unitary.
    \item If $d=1$ $\mathcal{R}$ sends the leakage  $l^{U_0,U_1}$ and $l^{\rho_{C}}$ to the Server.
    \item Case 1 ($c=0$):
    \begin{enumerate}
        \item $\mathcal{R}$ responds with $U_i(\rho_{C})$ to the Client.
    \end{enumerate}
    Case 2 ($c=1$):
    \begin{enumerate}
        \item The Server sends $(\mathcal{E}, \rho_{S})$ to $\mathcal{R}$.
        \item $\mathcal{R}$ sends $\mathcal{E}(i, \rho_{C}, \rho_{S})$ to the Client.
    \end{enumerate}
\end{enumerate}
\end{algorithmic}
\end{resource}

In our case the leakage $l^{U_0,U_1}$ is defined to be the full description of $U_0$ and $U_1$ in terms of the default angles, resource graphs and g-flows. It also includes the length of the input and the output for both unitaries. The other part of the leakage: $l^{\rho_{C}}$ is equal to the state of the input qubits that are public. The exact format of the deviation $(\mathcal{E}, \rho_{S})$ is described in section \ref{sec:proof}. \\

\begin{figure}[ht]
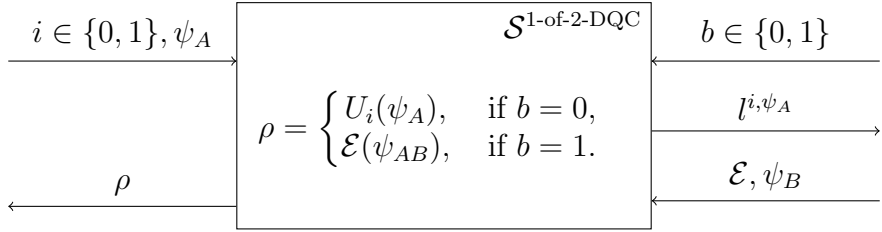

\centering
    \begin{bbrenv}{resource}
    \begin{bbrbox}[name=$\mathcal{S}^\text{1-of-2-DQC}$, minheight=3cm]
        $\rho = \left\{ \begin{matrix}
             U_i(\psi_A), & \text{ if } b=0, \\
            \mathcal{E}(\psi_{AB}), & \text{ if } b=1.
        \end{matrix} \right.$
    \end{bbrbox}
    \bbrmsgspace{6mm}
    \bbrmsgto{top={$i\in\bin, \psi_A$}, length=3cm}
    \bbrmsgspace{12mm}
    \bbrmsgfrom{top={$\rho$}, length=3cm}
    \bbrqryspace{6mm}
    \bbrqryfrom{top={$b\in\bin$}, length=3cm}
    \bbrqryspace{2mm}
    \bbrqryto{top={$l^{i,\psi_A}$}, length=3cm}
    \bbrqryspace{2mm}
    \bbrqryfrom{top={$\mathcal{E},\psi_B$}, length=3cm}
    \end{bbrenv}
\caption{1-of-2 DQC Resource. On the Client's interface, the Resource receives the input to the computation and sends the results. The Server's interface is inactive unless the Server decides to deviate maliciously from the protocol. If the latter is the case, the Server receives the leakage $l^{i,\psi_A}$ from the resource and submits its deviation $(\mathcal{E}, \psi_B)$.}
\label{fig:if}
\end{figure}

\begin{remark}
If the protocol constructs the 1-of-2 DQC resource it provides perfect blindness for the Clients input against an unbounded adversary.
\end{remark}

\subsubsection{Proof of Blindness.}
\label{sec:proof}

\begin{theorem}
\label{thm:blind}
    The Protocols \ref{prot:full masking pp} and \ref{prot:full masking cp} perfectly construct the 1-of-2 DQC Resource \ref{res:ubqc short} from a quantum channel.
\end{theorem}

\begin{proof}
In what follows we are going to describe a sequence of reductions and mention in every line which steps of the Protocols \ref{prot:full masking pp} and \ref{prot:full masking cp} are being replaced. All lines that are not mentioned stay the same.

In Reduction \ref{red:epr-pairs} we replace the masked qubits which the Client prepares for the Server with EPR-pairs. One half of every pair the Client sends to the Server and the other half they keep as their state.

This approach allows us to delay the choice of secret parameters $\theta$ and $r$ until later in the protocol. When the Server is requesting a measurement angle for one of those qubits the Client teleports there a correct value depending on the type of the qubit.

There are a few different cases to consider. Firstly, if the qubit is a bridge/break middle vertex the Client either teleports a $|0\ket/|1\ket$ state or $|+\ket/|-\ket$ to this node. This kind of qubits will be measured with angle $\frac{\pi}{2}$. If the qubit has to be deleted the Client teleports a $|0\ket/|1\ket$  state there as well. For these deleted qubits the Client also samples uniformly at random the parameters $\theta_v$ and $r$ and asks the Server to measure it with angle $\delta = \varphi' + \theta_v + r\pi$. For the rest of the qubits the Client samples $\theta_v$ from the set associated to this node defined in Section \ref{sec:full_protocol} for this node and teleports the value $|+_{\theta_v+ r\pi}\ket$ to the qubit. Then this qubit is measured with an angle $\delta = \varphi' + \theta_v$. It is equivalent to measuring $|+_{\theta_v}\ket$ with an angle $\delta = \varphi' + \theta_v + r\pi$ as in the original protocol.

It is easy to see that if the new parameters follow the uniform distributions as in the original protocol other steps of the computation are not affected by our modification. If for $\theta_v$ it is evident that the distribution is uniform we have to argue why $r$ is also uniformly random when it is recovered as a measurement outcome. 

The Server can manipulate their part of the EPR-pair which introduces changes to the Client's state. However Client's actions on their half of the pair do not depend on the messages sent earlier in the Protocol. It means that the situation is equivalent to not communicating at all. Hence, we can apply the No-Communication Theorem to justify why local operations on an entangled state do not transmit information. By no-communication theorem, all the measurement results that the Client obtains have the same distribution as if they shared a $\frac{1}{\sqrt{2}}(|+_{-\theta_v}\ket |+_{\theta_v}\ket + |-_{-\theta_v}\ket |-_{\theta_v}\ket)$ state with the Server. Therefore, when the Client measures their state in $|+_{-\theta_v}\ket, |-_{-\theta_v}\ket$ basis they are equally likely to obtain $0$ or $1$.

A similar argument holds to prove that the EPR half sent by the Client is indistinguishable from a correctly prepared $|+_{\theta_v}\ket$ state, $|0\ket/|1\ket$ state or a $|+\ket/|-\ket$ state. Consider the following two situations. In the first one the Client simply sends the EPR half to the Server, and in the second one they measure their qubit to teleport a correct value to the Server's half and send it to the Server. Due to No-Communication Theorem the Server cannot distinguish between these two cases, so all the measurements the Server can do on their EPR half will give the same results as when the qubit was prepared correctly.

In conclusion we have the following:
\begin{itemize}
    \item The parameters $\theta_v$, $r$ and $\delta$ are uniformly random from the corresponding sets. Hence, indistinguishable from original parameters.
    \item The operation order is different only on the side of the Client.
    \item The qubits that the Client sends to the Server are indistinguishable from correctly prepared qubits.
    \item After the modified steps the computation is exactly the same as in the original protocol for the secret parameters that we calculated.
\end{itemize}
    That is why the protocol after the modification is perfectly equivalent to the original protocol from the point of view of the Distinguisher.

\begin{reduction}
\caption{Using EPR pairs.}
\label{red:epr-pairs}
\begin{algorithmic}[0]
\STATE \begin{enumerate}
    \item (Lines \ref{pp:step4b} - \ref{pp:step4f} of PP.)
    \begin{enumerate}
        \item Public input qubits are sent to the Server.
        \item For every other qubit that should be sent to the Server, the Client prepares an EPR-pair and sends one half to the Server.
    \end{enumerate}
    \item (Line \ref{cp:step1} of CP.)
    Break operation: the Client measures their half of the EPR-pair in the computational basis. If the outcome equals $b$ - the Server has $|b\ket$ state.
    Bridge operation: they need to apply $H$ to their half of the EPR pair and measure it in the computational basis. If the measurement outcome is equal to $r$, the Server holds $\frac{1}{\sqrt{2}}(|0\ket + (-1)^{r}|1\ket)$ state.
    \item (Lines \ref{cp:step2b2} of CP.)
    \begin{enumerate}
        \item The Client samples $\theta_v$ from the same set as in the original protocol.
        \item (for deleted qubits) The Client measures their half of the EPR-pair in the computational basis. If the outcome equals b - the Server has $|b\ket$ state. The Client sends the measurement angle for $v$: $\delta = \varphi' + \theta_v + r\pi$, where $r \sample \{0,1\}$.
        \item \label{red1:step3c}(for non-deleted qubits) The Client measures the EPR pair that corresponds to $v$ in basis $|+_{-\theta_v}\ket, |-_{-\theta_v}\ket$ with outcome $r$. Due to the teleportation principle the other half is now in the state $|+_{\theta_v + r\pi}\ket$. The Client sends the measurement angle for $v$: $\delta = \varphi' + \theta_v$.
    \end{enumerate}
\end{enumerate}
\end{algorithmic}
\end{reduction}

Let us now discuss Reduction \ref{red:uni angles}. Here we let the Client first sample a random measurement angle $\delta$ and then calculate which secret parameters it would correspond to. This way, when we substitute the Client with the Simulator, the actions of the Simulator will not require access to Client's input or their secret keys.

To start, we note that the angles $\delta$ follow the same distribution as in the previous protocol. Let us now prove that we can change the order of the operations.

Suppose we are treating a non-deleted node. The angle $\theta_v = \delta - \varphi' - r\pi$ is calculated using the default angle, the angle $\delta$ that we just sent to the Server and the measurement outcomes we have already received from the Server. Hence, $\theta_v$ can be deduced up to a $\pi$ rotation.

The operations we perform to our half of the EPR pair teleport a $|+_{\theta_v + r\pi}\ket$ state to the Server's half. This qubit has already been measured, but the measurement is an operator of the form $(M(\delta) \otimes Id)(\rho_{1,2})$, where $\rho_{1,2}$ is the current state of the EPR pair. Therefore, it commutes with any operation we apply on the second qubit including $Id \otimes H Z(\theta_v)$. It means that our actions in the Reduction \ref{red:uni angles} are equivalent to first teleporting the correct value into the qubit and then performing the measurement (as we do in Reduction \ref{red:epr-pairs}). The same proof applies to the deleted nodes and the bridge/break nodes.

We also make modifications of another kind. For the classically masked angles in lines \ref{cp:step2b3} and \ref{cp:step2c} of CP we replace $\delta = \varphi' + r\pi$ with $\delta = \varphi + r\pi$. The angle $\varphi$ is no longer updated with respect to the $X$ and $Z$ corrections. We know that if this angle is masked classically then $(\varphi'- \varphi) \in \{0, \pi\} \pmod{2\pi}$. So by adding $r\pi$ with a uniformly random $r$ in both cases we make $\delta$ have the same distribution before and after the modification. An updated procedure for dealing with the measurement outcome described in step \ref{red2:step3c} of Reduction \ref{red:uni angles} allows us to recover the correct result of the measurement. In conclusion, the two reductions are perfectly indistinguishable.

\begin{reduction}[htp]
\caption{Uniform measuring angles.}
\label{red:uni angles}
\begin{algorithmic}[0]
\STATE \begin{enumerate}
    \item (Line \ref{cp:step1} of CP.)
    \begin{enumerate}
        \item The Server measures all bridge/break vertices with angle $\frac{\pi}{2}$
        \item Break operation: the Client measures their half of the EPR-pair in the computational basis. If the outcome equals $b$ - the Server was holding $|b\ket$ state.
            Bridge operation: they apply $H$ to their half of the EPR pair and measure it in the computational basis. If the measurement outcome is equal to $r$, the Server was holding $\frac{1}{\sqrt{2}}(|0\ket + (-1)^{r}|1\ket)$ state.
    \end{enumerate}

    \item (Line \ref{cp:step2b2} of CP.)
    \begin{enumerate}
        \item The Client samples $\delta$ from the set of possible angles for $v$ and sends $\delta$ as a measurement angle for $v$.
        \item The Client computes $\theta_v = \delta - \varphi'$.
        \item (for deleted qubits)  They measure their half of the EPR-pair in the computational basis. If the outcome equals b - the Server had $|b\ket$ state. The measurement result does not depend on the Z-rotations so we can set $r$ equal a random bit.
        \item (for non-deleted qubits) They measure the qubit in $|+_{-\theta_v}\ket, |-_{-\theta_v}\ket$ with the outcome recorded into the variable $r$.
    \end{enumerate}
    \item (Lines \ref{cp:step2b3} and \ref{cp:step2c} of CP.)
    \begin{enumerate}
        \item The Client sends $\delta = \varphi + r\pi$ to the Server with $r \sample \{0,1\}$. Note that now we are using the default angle $\varphi$.
        \item The Server returns a measurement outcome $s$.
        \item \label{red2:step3c} The Client computes $\varphi'$ and if $\varphi + r\pi = \varphi' + \pi$ the Client flips $s$ and keeps it for the future corrections; otherwise they keep $s$ as it was.
    \end{enumerate}
\end{enumerate}
\end{algorithmic}
\end{reduction}

As the last part of the proof we construct the Simulator that makes the Real world and the Ideal world protocol executions perfectly indistinguishable. Intuitively, in the ideal world we push the measurements of the Client's EPR halves to the very end of the protocol and delegate them to the Ideal Resource that knows Client's inputs.

With the information from the Simulator and the EPR halves the Ideal Resource reconstructs the output that the Client would have obtained if they ran the UBQC protocol with the Adversary in the real world. The set of Client's secret keys for this execution is decided by the Ideal Resource.

Here we give the definition of the Simulator \ref{sim: ubqc} , the description of the deviation provided to the Resource and which computation the Resource performs to reconstruct the output.

\begin{simulator}[htp]
\caption{The UBQC protocol simulator.}
\label{sim: ubqc}
\begin{algorithmic}[0]
\STATE \begin{enumerate}
    \item For every public input qubit the Simulator transmits it to the Server in clear from the leakage.
    \item For every other qubit that has to be sent to the Server the Simulator prepares an EPR pair and sends one half to the Server.
    \item For every masked qubit $v$ the Simulator does the following:
    \begin{enumerate}
        \item If $v$ belongs to the case \ref{cp:step2b2} of CP the Simulator samples $\delta$ from a set of possible measurement angles for $v$ and sends $\delta$ as a measurement angle for $v$.
        \item Else if $v$ belongs to case \ref{cp:step2b3} of CP Simulator sends $\delta = \varphi + r\pi$ to the Server with $r \sample \{0,1\}$.
        \item Else if $v$ belongs to case \ref{cp:step2b4} of CP Simulator sends $\delta = \varphi'$ corresponding to the published measurement outcomes.
        \item The Simulator receives $s \in \{0,1\}$ from the Server.
    \end{enumerate}
    \item The Simulator receives the output qubits from the Server.
    \item In the end the Simulator sends to the 1-of-2 DQC Resource the following:
    the generated half EPR pairs, the output qubits, all measurement outcomes $s$, angles $\delta$ for the masked qubits, bits $r$ from the case \ref{cp:step2b3}.
\end{enumerate}
\end{algorithmic}
\end{simulator}

\begin{resource}[htp]
\caption{1-of-2 Delegated Quantum Computation Resource. Full description.}
\label{res:ubqc comp}
\begin{algorithmic}[0]
\STATE \begin{enumerate}
    \item The Resource receives from the Client the input qubits and the bit $i$ that indicates which unitary they would like to perform.
    \item From the Server's interface the Resource receives the information describing the deviation in the following form: a qubit for every masked qubit sent to the Server in the preparation phase, the output qubits, all measurement outcomes $s$, angles $\delta$ for the masked qubits, bits $r$ from the case \ref{cp:step2b3}.




    \item For every bridge/break qubit:
    \begin{enumerate}
        \item Break operation: the Resource measures the corresponding
        in the computational basis and record the outcome $b$.
        \item Bridge operation: the Resource applies $H$ to the corresponding qubit and measures it in the computational basis and records the measurement outcome in $r$.
    \end{enumerate}
    \item Then for every masked qubit the Resource does the following:
    \begin{enumerate}
        \item The Resource computes the correct $\varphi'$ from the measurement outcomes received on the Server's interface and the correct measurement outcome the Resource recomputed.
        \item The Resource computes $\theta_v = \delta - \varphi'$.
        \item \label{res:step1} (for deleted nodes) The Resource measures the corresponding in the computational basis and records the outcome $b$. The resource also samples $r$ at random. If $r = 1$ the Resource flips the measurement outcome $s$.
        \item \label{res:step2} (for non-deleted nodes) The Resource measures the corresponding qubit in $|+_{-\theta_v}\ket, |-_{-\theta_v}\ket$ basis with outcome $r$. If $r = 1$ the Resource flips the measurement outcome $s$.
        \item \label{res:step3} (for classical masking nodes) If $\delta = \varphi' + \pi$ the Resource flips the measurement outcome of this node.
    \end{enumerate}
    \item The Resource now computes the final correction, applies it to the output qubits and outputs these qubits at the Client's interface.
\end{enumerate}
\end{algorithmic}
\end{resource}

It is easy to see that the messages the Simulator sends to the Adversary (run by a Distinguisher) have the same distributions as the messages after the Reduction \ref{red:uni angles}.

In the description of the Resource \ref{res:ubqc comp} the reader can see how the information which the Simulator obtains from the Adversary can be used to reconstruct real UBQC computation results including the deviation of the Adversary.

To see the correctness of this computation we argue that combining the Simulator and the Resource into one party that executes these two protocols sequentially is equivalent to the original protocol after the Reduction \ref{red:uni angles}.

When the Resource reconstructs the computation its main objective is to recover the correct measurement outcomes one by one to know how to update the following angles and the following outcomes. Similarly to the Reduction \ref{red:uni angles} when treating a particular node the Resource has already computed a correct measurement outcome for every node in its past so now the Resource can determine $\varphi'$ for the current node and follow the instructions in steps  \ref{res:step1} - \ref{res:step3} to compute the correct value of $s$.

We proved that every modification of the protocols was perfectly indistinguishable. Therefore, the Protocol \ref{prot:full masking pp} and the Protocol \ref{prot:full masking cp} $\epsilon$-construct the Ideal Resource defined in Resource \ref{res:ubqc short} and Resource \ref{res:ubqc comp} for $\epsilon =0$. This equivalence implies perfect blindness for the Clients input $i$ and the private input qubits.
\end{proof}

\section {Conclusion and Future Directions}

In this work, we have explored the fundamental question of whether it is possible to reduce the quantum communication complexity in universal blind quantum computing. Through a sequence of impossibility results, we showed that no quantum process, whether exact, approximate, separable or entangled, performed solely on the server side can distribute or expand encrypted resource states without violating blindness. These results go beyond the no-cloning theorem, revealing that even approximate distributor machines introduce dependencies exploitable via differential attacks, thereby breaking security. The linearity of quantum mechanics, while being the root of both no-cloning and our no-distribution theorem, is not sufficient to reduce our result to the former, as the randomness distribution task is more general. 

As a constructive result, we introduced a new paradigm we call selectively blind quantum computing (SBQC), which considers the case where a client wishes to delegate one computation from a known set, without revealing which one has been chosen. We formalized this in the language of composable cryptography and presented protocols that reduce the required number of quantum communications depending on the differences between the computation graphs required to implement various computations in the target set. The core idea is that only certain parts of the computation need to be hidden. We presented masking techniques for both angle differences and structural differences using a merged graph formalism.

A key insight from our protocol is the direct relationship between the placement of non-Clifford gates and the cost of hiding. In particular, non-Clifford gates are the ones that depending on their location within the circuit they might require masking via additional qubits sent from the client. If a circuit is compiled in such a way that these non Clifford gates are located in non-overlapping regions of the propagation cones of client's secret parameter then less masking is necessary leading to reduced quantum communication. This provides a new perspective on circuit compilation, distinct from existing optimization strategies for fault-tolerant computing (e.g., minimizing T-count or T-depth). Our approach introduces a new objective: structuring circuits for optimal hiding, opening a new direction in quantum compilation informed by security rather than error correction.

 Although we focused exclusively on Blind quantum computing in this work, it would be important to explore whether these hiding techniques can be adapted to verifiable blind quantum computation protocols. Whether the same trade-offs apply, or new impossibility results arise in that context, remains an open problem. A crucial component of the proof of verifiability is Pauli twirling that cannot be applied to our protocol if some nodes do not have any masking. But in the private input scenario of our protocol, every qubit-masked node is masked by encrypting the measurement angle, and every other node is masked with an $r\pi$ rotation. One future direction would be to investigate whether this masking is sufficient for the Pauli twirling and making our protocol verifiable.

It is also natural to ask whether the notion of SBQC can scale efficiently beyond two computations, and what limitations emerge when attempting to balance leakage with communication cost. For example it is not clear what is the best strategy for merging two given graphs in an optimised way. Having this algorithm would allow us to estimate the improvement that our protocol gives in terms of the size of the entangled state of the Server side. More precisely, we expect a trade-off between the size of the resource graph and the qubit communication complexity. Although this problem seems difficult to solve for arbitrary graphs, it would already be very useful to find an optimal merger for $\mathbb{Z}_2$ lattices of different size or other regular graphs used in practice. Altogether, our work both constrains and extends the design space for blind quantum computation and introduces new tools to navigate it.

\section{Acknowledgments}
MD, DL, and EK acknowledge support from the Quantum Advantage Pathfinder (QAP) research program within the UK’s National Quantum Computing Center (NQCC). MD and EK also acknowledges the Integrated Quantum Networks Hub, grant reference EP/Z533208/1. AP and MD would like to thank Robert Booth for useful discussions in the early stages of the project. O.L would like to thank the members of the crypto group in LIP6 for helpful discussions.  

\bibliography{one-of-two_ubqc}
\bibliographystyle{ieeetr}

\clearpage
\appendix

\section{Alternative Proof of Lemma~\ref{lemma:no-sep-gadget-d}}
\label{sec:alternative-proof-lemma-1}

Here, we present an alternative proof of Equation~\eqref{eq:d-gadget} to highlight the connection between the lemma and the linearity of quantum mechanics
\begin{proof}
Assume the existence of an isometry as described. The output of this isometry takes the form $|+{\theta_1}\rangle \otimes |+{\theta_2}\rangle$. Expressing each state in the Hadamard basis, the resulting state can be written as: $|\psi\rangle = |+_{\theta_1} \rangle \otimes |+_{\theta_2} \rangle= a|++\rangle + b |+-\rangle + c|-+\rangle + d |--\rangle$ where the coefficients $a, b, c, d$ satisfy the normalization condition. This implies that, by varying the input parameter $\theta$, the isometry's output spans the full four-dimensional Hilbert space of two qubits: 
\begin{equation}
    D |+_{\theta} \rangle \in \operatorname{span} \{ {|++\rangle , |+-\rangle , |-+\rangle , |--\rangle } \}
\end{equation}
Now, consider an alternative perspective on the output dimensionality. Decompose the input state as $|+_{\theta}\rangle = \alpha |+\rangle + \beta |-\rangle$. Using Equation~\eqref{action on basis} and the linearity of quantum mechanics, the output of the isometry becomes: 
\begin{equation}
    D |+_{\theta} \rangle = D (\alpha |+\rangle + \beta |-\rangle)= \alpha |++\rangle + \beta_2 |-+\rangle + \gamma_2 |--\rangle \in \operatorname{span} \{ {|++\rangle , |-+\rangle , |--\rangle} \}
\end{equation}
where $\beta_2$ and $\gamma_2$ depend on the isometry's parameters.

A contradiction arises between two ways of evaluating the dimensionality of the output. The former suggests the output spans a four-dimensional space, while the latter indicates it is confined to a three-dimensional subspace. This inconsistency demonstrates a conflict between the assumed existence of the isometry in Equation and the linearity of quantum mechanics. 

\noindent Furthermore, the lemma's validity extends beyond inputs of the form $|+_{\theta}\rangle$. Both proofs' mathematical arguments hold even when the input is a generic one-qubit state, rather than a planar state. If we generalize the input condition to an arbitrary state, the lemma's statement depends not on the number of parameters characterizing the input and output states, but on their functional dependence. Specifically, the condition in Equation~\eqref{action on basis} would be modified such that $D |-\rangle = |-\rangle \otimes |\phi\rangle$ replaces the second condition. This alteration again results in a dimensionality reduction, as the basis state $|+-\rangle$ is excluded from the output span.

Moreover, extending the output to multiple registers — e.g. requiring a pair of two-qubit product states (where each may be entangled) — is similarly unattainable. The same dimensionality arguments preclude achieving such an output configuration. Thus, no such isometry exists. 
\end{proof}

\end{document}